\documentclass[preprint,3p,times]{elsarticle}

\usepackage{amsmath}
\usepackage{amssymb}
\usepackage{amsthm}
\usepackage{mathtools}
\usepackage{tensor}
\usepackage{multirow}
\usepackage{enumerate}
\usepackage{graphicx}
\usepackage{color}

\usepackage{hyperref}

\newcommand{\splitcelltab}[1]{\begin{tabular}{@{}c@{}}#1\end{tabular}} %use \\ to split

\newcommand{\wrt}{w.\,r.\,t.}
\newcommand{\eg}{e.\,g.}
\newcommand{\st}{s.\,t.} 
\newcommand{\ie}{i.\,e.}

\newcommand{\etc}{etc.}
\newcommand{\resp}{resp.}

\newcommand{\formComma}{\,\text{,}}
\newcommand{\formPeriod}{\,\text{.}}

\newcommand{\surf}{\mathcal{S}}

\newcommand{\R}{\mathbb{R}}

\newcommand{\tangent}[2][]{\tensor{\operatorname{T}\!}{#1}#2}
\newcommand{\tangentS}[1][]{\tangent[#1]{\surf}}
\newcommand{\tangentR}[1][]{\tangent[#1]{\R^3\vert_{\surf}}}
\newcommand{\tangentQS}{\tensor{\operatorname{Q}\!}{^2}\surf}
\newcommand{\tangentQR}{\tensor{\operatorname{Q}\!}{^2}\R^3\vert_{\surf}}
\newcommand{\tangentCQR}{\tensor{\operatorname{C_{\surf}Q}\!}{^2}\R^3\vert_{\surf}}
\newcommand{\tangentAR}{\tensor{\operatorname{A}\!}{^2}\R^3\vert_{\surf}}
\newcommand{\tangentAS}{\tensor{\operatorname{A}\!}{^2}\surf}
\newcommand{\tangentSymS}{\tensor{\operatorname{Sym}\!}{^2}\surf}

\newcommand{\tangentStar}[2][]{\tensor*{\operatorname{T}\!}{#1}#2}

\newcommand{\paraC}{X}
\newcommand{\para}{\boldsymbol{\paraC}}
\newcommand{\normalC}{\nu}
\newcommand{\normal}{\boldsymbol{\normalC}}
\newcommand{\shopC}{I\!I}
\newcommand{\shop}{\boldsymbol{\shopC}}
\newcommand{\meanc}{\mathcal{H}}
\newcommand{\gaussc}{\mathcal{K}}

\newcommand{\landau}{\mathcal{O}}

\DeclareRobustCommand{\GGamma}{\text{\raisebox{\depth}{\scalebox{1}[-1]{$\mathbb{L}$}}}}

\newcommand{\Tr}{\operatorname{Tr}}
\renewcommand{\div}{\operatorname{div}}

\newcommand{\proj}{\operatorname{\Pi}}
\newcommand{\projS}[1][]{\proj_{\tangentS[#1]}}
\newcommand{\projQS}{\proj_{\tangentQS}}
\newcommand{\projCQR}{\proj_{\tangentCQR}}

\newcommand{\jau}{\mathcal{J}}

\newcommand{\timeJ}{\mathfrak{J}}
\newcommand{\timeLuu}{\mathfrak{L}^{\sharp\sharp}}
\newcommand{\timeLll}{\mathfrak{L}^{\flat\flat}}

\newcommand{\timeLu}{\mathfrak{L}^{\sharp}}
\newcommand{\timeLl}{\mathfrak{L}^{\flat}}

\newcommand{\Dt}[1][]{\operatorname{D}^{#1}_t}
\newcommand{\Dmat}{\Dt[\,\mfrak]}
\newcommand{\Dupp}{\Dt[\sharp]}
\newcommand{\Dlow}{\Dt[\flat]}
\newcommand{\Djau}{\Dt[\jau]}
\newcommand{\DCQmat}{\Dt[\operatorname{C}_{\surf},\mfrak]}

\newcommand{\dt}{\operatorname{d}_t}

\newcommand{\potenergy}{\mathfrak{U}}

\newcommand{\hil}{\operatorname{L}^{\!2}}
\newcommand{\hilspace}[1]{\hil(#1)}
\newcommand{\inner}[2]{\left\langle #2 \right\rangle_{#1}}
\newcommand{\normsq}[2]{\left\| #2 \right\|_{#1}^2}
\newcommand{\innerH}[2]{\inner{\hilspace{#1}}{#2}}
\newcommand{\normHsq}[2]{\normsq{\hilspace{#1}}{#2}}

\newcommand{\overdot}[1]{\dot{\overline{#1}}}
\newcommand{\overdotinner}[2]{\overdot{\inner{}{#2}}_{#1}} % ensures that the space spec is outside overdot

\newcommand{\gb}{\boldsymbol{g}}

\newcommand{\qb}{\boldsymbol{q}}
\newcommand{\Qb}{\boldsymbol{Q}}
\newcommand{\rb}{\boldsymbol{r}}
\newcommand{\Rb}{\boldsymbol{R}}
\newcommand{\pb}{\boldsymbol{p}}
\newcommand{\Pb}{\boldsymbol{P}}
\newcommand{\eb}{\boldsymbol{e}}
\newcommand{\Eb}{\boldsymbol{E}}
\newcommand{\bb}{\boldsymbol{b}}
\newcommand{\Gb}{\boldsymbol{G}}
\newcommand{\Sb}{\boldsymbol{S}}
\newcommand{\Ab}{\boldsymbol{A}}

\newcommand{\Psib}{\boldsymbol{\Psi}}

\newcommand{\Vb}{\boldsymbol{V}}
\newcommand{\vb}{\boldsymbol{v}}
\newcommand{\vnor}{v_{\bot}}
\newcommand{\lambdab}{\boldsymbol{\lambda}}
\newcommand{\ub}{\boldsymbol{u}}

\newcommand{\etab}{\boldsymbol{\eta}}

\newcommand{\Omegab}{\boldsymbol{\Omega}}

\newcommand{\sbb}{\boldsymbol{s}}

\newcommand{\chib}{\boldsymbol{\chi}}

\newcommand{\Gbcal}{\boldsymbol{\mathcal{G}}}
\newcommand{\Abcal}{\boldsymbol{\mathcal{A}}}
\newcommand{\Qbcal}{\boldsymbol{\mathcal{Q}}}

\newcommand{\Id}{\boldsymbol{Id}}
\newcommand{\IdS}{\Id_{\surf}}

\newcommand{\rot}{\operatorname{rot}}

\newcommand{\Comp}{\mathrm{C}}

\newcommand{\DeltaS}{\operatorname{\Delta}_{\Comp}}
\newcommand{\nablaS}{\operatorname{\nabla}_{\Comp}}

\newcommand{\nablahat}{\operatorname{\widehat{\nabla}}}

\newcommand{\DeltaCQS}{\operatorname{\Delta}_{\Comp}^{\operatorname{C}_{\surf}}}

\newcommand{\mfrak}{{\!\mathfrak{m}}}
\newcommand{\ofrak}{{\!\mathfrak{o}}}

\newcommand{\dS}{\operatorname{d}\surf}

\newtheorem{theorem}{Theorem}
\newtheorem{lemma}[theorem]{Lemma}
\newtheorem{corollary}[theorem]{Corollary}

\journal{Journal of Geometry and Physics}

\title{Tensorial time derivatives on moving surfaces: General concepts and a specific application for surface Landau-de Gennes models}

\author[1]{Ingo Nitschke\corref{cor1}}
\cortext[cor1]{Corresponding author: ingo.nitschke@tu-dresden.de}
\author[1,2,3]{Axel Voigt}

\address[1]{Institut f{\"u}r Wissenschaftliches Rechnen, Technische Universit{\"a}t Dresden, 01062 Dresden, Germany}
\address[2]{Dresden Center for Computational Materials Science (DCMS), Technische Universit{\"a}t Dresden, 01062 Dresden, Germany}
\address[3]{Center for Systems Biology Dresden (CSBD), Pfotenhauerstr. 108, 01307 Dresden, Germany}

\begin{document}

\begin{frontmatter}

\begin{abstract}
Observer-invariance is regarded as a minimum requirement for an appropriate definition of time derivatives. We systematically discuss such time derivatives for surface tensor field and provide explicit formulations for material, upper-convected, lower-convected and Jaumann/corotational time derivatives which all lead to different physical implications. We compare these results with the corresponding time derivatives for tangential tensor fields. As specific surface 2-tensor fields we consider surface Q-tensor fields and conforming surface Q-tensor fields and apply the results in surface Landau-de Gennes models for surface liquid crystals. 
\end{abstract}


\end{frontmatter}

\section{Introduction}

Observer-invariant time derivatives for tensor-fields on moving surfaces $\surf \subset \R^3$ are important ingredients for various applications, such as fluid deformable surfaces and surface liquid crystals, see, e.g., \cite{Torres-SanchezMillanArroyo_JoFM_2019,reuther2020numerical,Krauseetal_arXiv_2022,nitschke2019hydrodynamic,Nitschkeetal_PRSA_2020}. They determine specific rates of change independently of their observation and specify transport mechanism reflecting a certain inertia in the considered quantity induced by material motions. For tangential tensor-fields, defined in the tangent bundle of $\surf$, denoted by $\tangentS[^n]$, time derivatives for arbitrary observer are discussed in detail in \cite{NitschkeVoigt_JoGaP_2022}. Unlike for scalar fields, where the time derivative is uniquely defined, severe differences in the evolution of the tangential tensor-field, e.g. a surface director field or a surface Q-tensor field in surface liquid crystal models \cite{NitschkeSadikVoigt_A_2022}, have been identified. The implications of these differences, e.g. in morphogenesis \cite{Maroudas-Sacks_NP_2021,Hoffmann_SA_2022,Morris_2022}, which can be modelled using fluid deformable surfaces and surface liquid crystals are not yet explored. The requirement of the surface tensor-fields to be tangential might be too strong for such applications. They require surface tensor-fields with an additional normal component, see, e.g., \cite{Bartelsetal_IFB_2012} for a director field on a flexible membrane and \cite{Golovatyetal_JNS_2017,Nitschke_2018,Nestler_2020,Boucketal_arXiv_2022} for surface Q-tensor-fields but on stationary surfaces. These surface tensor-fields are defined in $\tangentR[^n]$, and need slightly different time derivatives, which respect the embedding space $\R^3$ as well as the surface $\surf$. We will systematically discuss these time derivatives, their properties and relations and apply them to a surface Landau-de Gennes model.

Unlike in \cite{NitschkeVoigt_JoGaP_2022}, we do not take a spacetime manifold as a basis for observer-invariant time derivatives. The advantage is an improved readability due to lesser abstract concepts. The disadvantage is that we do not get any longer observer-invariance for free as a results of a covariance principle \wrt\ the choice of spacetime coordinates. The main issue to develop observer-invariant time derivatives is that the time $t$ is not a coordinate of $\surf$, but rather a parameter to describe time-dependencies \wrt\ an observer and the relation between time and space. We need to demonstrate that the time derivatives of tensor-fields on moving surfaces are invariant within the observer class depicting the moving surface. In parts we circumvent this issue by stipulating time derivatives for a material observer (Lagrange perspective) and transform these representation to an arbitrary observer. We only consider instantaneous tensor-fields. 

We introduce notation in subsection \ref{sec:notation} and provide a short tabular summary in subsection \ref{sec:summary}. The actual derivation of time derivatives is constituted in section \ref{sec:derivation}, which is organized in the following way.
Subsection \ref{sec:general_approach} describes the general approach to obtain time derivatives.
Basically, we use a differential quotient \wrt\ the time parameter $ t $ \st\ a time derivative yields a certain rate of the considered tensor-fields.
Such an approach is only sufficient for fixed choices of convenient pullbacks, which are capable of evaluating a ``future'' tensor-field on the current surface. We illustrate this approach for scalar fields in $ \tangentS[^0] $, where such a pullback seems to be uniquely given.
The situation changes for $ n $-tensor fields with $ n\ge 1 $, where different pullbacks lead to different time derivatives.
In subsection \ref{sec:Dmat} we derive the material time derivative,
in subsection \ref{sec:Dupp} the upper-convected time derivative, 
in subsection \ref{sec:Dlow} the lower-convected time derivative
and in subsection \ref{sec:Djau} the Jaumann/corotational time derivative.
The individual derivatives are build on each other. 
In all of these subsections we also consider vector- as well as 2-tensor-fields separately for the sake of readability. Additionally, we show at the end of each subsection that all of these time derivatives are thin-film limits of usual flat $ \R^3 $ time derivatives. With surface Landau-de Gennes models for surface liquid crystals in mind \cite{Nestler_2020}, we treat in 
subsection \ref{sec:qtensor} Q-tensor fields in $ \tangentQR $, which are symmetric and trace-free, as a special case of surface 2-tensor fields.
Here we consider only the material and Jaumann/corotational time derivative.
Moreover, we discuss surface conforming Q-tensor fields in $ \tangentCQR $, where the eigenvector spaces are aligned to the surface.
This gives the opportunity to modify the material derivative to a simpler representation. Using these tools we formulate surface Landau-de Gennes models on evolving surfaces which lead for the material time derivative to the same formulation as postulated in \cite{Nestler_2020}.

\subsection{Notation}
\label{sec:notation}

\begin{table}
\centering
\renewcommand{\arraystretch}{1.4}
\begin{tabular}{|ll|}
 \hline
 $ \tangentS[^0] = \tangentR[^0]  $ & scalar fields \\
 \hline\hline
 $ \tangentR  $ & vector fields \\
 \hline
 $ \tangentS < \tangentR $ & tangential vector fields \\
 \hline\hline
 $ \tangentR[^2] $ & 2-tensor fields \\
 \hline
 $ \tangentS[^2] < \tangentR[^2]$ & tangential 2-tensor fields \\
 \hline 
 $ \tangentQR < \tangentR[^2] $ & Q-tensor fields (trace-free, symmetric) \\
 \hline
 $ \tangentCQR < \tangentQR $ & surface conforming Q-tensor fields (only normal and tangential eigenvectors) \\
 \hline
 $ \tangentQS < \{\tangentCQR, \tangentS[^2] \} $ & tangential Q-tensor fields (trace-free, symmetric) \\
 \hline
\end{tabular}
\caption{Most used tensor field spaces and their local subtensor space relations.
        All Q-tensor related spaces are defined in section \ref{sec:qtensor}}
\label{tab:most_used_spaces}
\end{table}

\begin{table}[t]
\centering
\renewcommand{\arraystretch}{1.4}
\begin{tabular}{|ll|}
\hline
$ \para_{\mfrak} $, $ \para_{\ofrak} $
    & material and observer parameterization\\
\hline
\splitcelltab{$ g_{\mfrak ij}=\inner{\tangentS}{\partial_i\para_{\mfrak}, \partial_j\para_{\mfrak}} $,\\ 
               $ g_{\ofrak ij}=\inner{\tangentS}{\partial_i\para_{\ofrak}, \partial_j\para_{\ofrak}} $}
    & material and observer metric tensor proxy field\\
\hline
$ g_{\mfrak}^{ij} $, $ g_{\ofrak}^{ij} $
    & matrix inverse of proxy fields $ g_{\mfrak ij} $ and $ g_{\ofrak ij} $\\
\hline
$ \Gamma_{\mfrak ij}^{k} $, $ \Gamma_{\ofrak ij}^{k} $
    & Christoffel symbols of 2nd kind \wrt\ $ g_{\mfrak ij} $ and $ g_{\ofrak ij} $\\
\hline
$ \normal\in\tangentR $
    & normal field \\
\hline
$ \shop = - \nablaS\normal $, $ \meanc = \Tr\shop $
    & shape operator and mean curvature\\
\hline
$ \Vb_{\mfrak} = \vb_{\mfrak} + \vnor\normal\in\tangentR $
    & material velocity \\
\hline
$ \Vb_{\ofrak} = \vb_{\ofrak} + \vnor\normal\in\tangentR $
    & observer velocity    \\
\hline
$ \ub = \Vb_{\mfrak} - \Vb_{\ofrak} \in\tangentS $
    & relative velocity \\
\hline
$ \Gb[\Vb] = \nabla\vb - \vnor\shop\in\tangentS[^2] $
    & tangential gradient of $ \Vb\in\tangentR $ \\
\hline
$ \bb[\Vb] = \nabla\vnor + \shop\vb\in\tangentS $
    & non-tangential gradient of $ \Vb\in\tangentR $ \\
\hline
$ \nablaS\Vb = \Gb[\Vb] + \normal\otimes\bb[\Vb] $
    & surface gradient of $ \Vb\in\tangentR $ \\
\hline
$ \Ab[\Vb] = \frac{\Gb[\Vb] - \Gb^T[\Vb]}{2} = \frac{\nabla\vb - (\nabla\vb)^T}{2} $
    & antisymmetric part of $ \Gb[\Vb] $ \\
\hline
$ \Gbcal[\Vb] = \nablaS\Vb - \bb[\Vb]\otimes\normal $
    & adjusted surface gradient of $ \Vb\in\tangentR $ \\
\hline
$ \Abcal[\Vb] = \frac{\Gbcal[\Vb] - \Gbcal^T[\Vb]}{2} $
    & antisymmetric part of $ \Gbcal[\Vb] $ \\
\hline
\end{tabular}
\caption{Some frequently used quantities. Note that $ \Vb $ is used as a placeholder for $ \Vb_{\mfrak} $ and $ \Vb_{\ofrak} $ in case the observer choice is relevant.}
\label{tab:quantities_overview}
\end{table}

We mainly adopt notations from \cite{NitschkeSadikVoigt_A_2022}.
Nevertheless, we give a condensed summary in this section including some notational extensions.
A moving surface $ \surf $ is sufficiently described by parameterizations
\begin{align}\label{eq:para}
    \para:\quad \mathcal{T}\times\mathcal{U} \rightarrow \R^3:\quad (t,y^1,y^2)\mapsto\para(t,y^1,y^2)\in\surf\vert_{t}\formComma
\end{align}
where $ \mathcal{U}\subset\R^2 $ is the chart codomain and $ \mathcal{T}=[t_0,t_1]\subset\R $ the time domain.
For simplicity we assume that $ \para(t,\mathcal{U})=\surf\vert_{t} $ can be achieved by a single time-depending parameterization $ \para $ for all $ t\in\mathcal{T} $.
The results can be extended to the more general case considering subsets providing an open covering of $ \surf $.
We omit the time parameter $ t $ in the notation if it is clear that the considered term can be evaluated temporally locally.
A parameterization is not uniquely given for a moving surface. 
For instance comprises $ \para $ information about the observer. Due to this, we subscribe quantities with $ \mfrak $ if we consider the material observer and 
$ \ofrak $ if we consider an arbitrary observer, see \cite{NitschkeVoigt_JoGaP_2022} for more details about observer.
%Note that  $ \boldsymbol{Z} $ is used for the parameterization instead of $ \para $ there.
One could refer the material observer to the Lagrangian perspective/specification.
Since we also consider motion in normal direction of the surface and the observer has to follow the material in this direction,
a pure Eulerian perspective does not exist on moving surfaces generally.
Note that we assume that $ \para $ provides a sufficiently smooth embedding of $ \surf  $ into $ \R^3 $.

We write $ \tangentR[^n] $ as a shorthand for the space of sufficiently smooth $ n $-tensor field on $\surf\subset\R^3$, 
\ie\ for $ \Rb\in\tangentR[^n] $ and $(y^1,y^2)\in\mathcal{U}$ the quantity $\Rb(y^1,y^2)\in\tangentStar[^n_{\para(y^1,y^2)}]{\R^3}\cong(\R^3)^n$
is a usual $ \R^3 $-$n$-tensor defined at $\para(y^1,y^2)\in\surf$.
This means that we handle tensor bundles and fields (section of bundles) synonymously due to the assumed smooth structure.
We also does not distinguish between co- and contravariant tensor fields, and everything between, in index-free notations, since
they are isomorph by the musical isomorphisms ($ \flat $, $\sharp$) for a given metric and all operators used in this paper respect that.   
The space of tangential $ n $-tensor fields $ \tangentS[^n] $ is a subtensor field of $ \tangentR[^n] $,
\ie\ it holds the subtensor relation $ \tangentStar[^n_{\para(y^1,y^2)}]{\surf} < \tangentStar[^n_{\para(y^1,y^2)}]{\R^3} $
for all $(y^1,y^2)\in\mathcal{U}$.
The space  $ \tangentS[^n] $ contains only the fields from $ \tangentR[^n] $ that can be represented by a tangential frame.
We summarize the most used subtensor fields of $ \tangentR[^n] $ for $ n\in{0,1,2} $ in this paper in table \ref{tab:most_used_spaces}.
Some of them are defined in their associated section, where they are used.
For $ n=1 $ we omit the index, \eg\ it is $ \tangentS = \tangentS[^1] $.
Every subtensor field relation brings its uniquely defined orthogonal projection $ \proj_{(\cdot)} $ along,
which is labeled by its image, \ie\ the subtensor field space.
The orthogonal projection $ \projS[^n]:\tangentR[^n] \rightarrow \tangentS[^n] $ projects $ n $-tensor fields into tangential $ n $-tensor fields for instance. 
We use the global Cartesian as well as local tangential frames and thus, for a better readability, also different index notations (Ricci calculus) in accordance with their frame.
We apply capital Latin letters $ A,B,C,\ldots $ \wrt\ the Cartesian frame $ \{\eb_A\}$,
\eg\ we could use $ R^{AB}\eb_A\otimes\eb_B $ to describe a 2-tensor field $ \Rb\in\tangentR[^2] $.
Small Latin letters $ i,j,k,\ldots $ are used \wrt\ the tangential frame $ \{\partial_i\para\} $ derived from parameterization \eqref{eq:para}.
For instance we could write $ r^{ij}\partial_i\para\otimes\partial_j\para $ for a tangential 2-tensor field $ \rb\in\tangentS[^2] $.
%On occasion capital Latin letters $ I,J,K,\ldots $ refer to the thin film frame $ \{\partial_1\para, \partial_2\para, \partial_\xi\para\} $
%\wrt\ parameterization \eqref{eq:para}.
%For example, the 2-tensor field $ \Rb $ above could also be written as $ R^{IJ}\partial_I\para\otimes\partial_J\para $, where we either just
%define $\partial_\xi\para := \normal$ as the normals field or by evaluating the thin film frame \ref{eq:tf_frame} at $\xi=0$.
%See \cite{Nitschke_2018} for more details on thin film notations.

We only use two kinds of spatial derivatives.
One is the covariant derivative $ \nabla:\tangentS[^n] \rightarrow \tangentS[^{n+1}] $ defined by the Christoffel symbols 
$ \Gamma_{ijk} = \frac{1}{2}( \partial_i g_{jk} + \partial_j g_{ik} - \partial_k g_{ij}) $ in a usual way, 
where $ g_{ij} $ is the covariant proxy of the metric tensor. 
In index notations, we represent $ \nabla $ with a stroke ``$ \vert $''.
For instance, we write $ \tensor{[\nabla\rb]}{^{ij}_k} = \tensor{r}{^{ij}_{|k}} = \partial_k r^{ij} + \Gamma_{kl}^i r^{lj} + \Gamma_{kl}^j r^{il} $
for $ \rb\in\tangentS[^2] $.
The other one is the surface derivative $ \nablaS:\tangentR[^n] \rightarrow \tangentR[^n]\otimes \tangentS < \tangentR[^{n+1}] $
defined as the covariant derivative on the Cartesian proxy components, which are scalar fields in $ \tangentS[^0] $.
As an example, it is $ \nablaS\Rb = \eb_A\otimes\eb_B\otimes\nabla R^{AB} $ valid for $ \Rb\in\tangentR[^2] $.
For readers from other communities, it holds $ \nablaS\Rb = (\nablahat\widehat{\Rb})\vert_{\surf} \IdS $,
where $ \widehat{\Rb}\in\tangent[^n]{\R^3} $ is an arbitrary smooth extension \st\ $ \widehat{\Rb}\vert_{\surf} = \Rb\in\tangentR $ is valid,
$ \nablahat:\tangent[^n]{\R^3} \rightarrow \tangent[^{n+1}]{\R^3} $ the usual $ \R^3 $-gradient
and $ \IdS\in\tangentS[^2] $ the surface identity, \resp\ tangential projection, tensor field, 
\eg\ given by $ [\IdS]^{AB} = \delta^{AB} - \normalC^A\normalC^{B} $ or $ [\IdS]^{ij} = g^{ij} $.  
Both derivatives are also related outside the Cartesian frame and we give these relations in the appropriated locations in this paper where they are needed.
Further definitions for covariant differential operators like $ \div $ (divergence), $ \rot $ (curl), $ \Delta $ (Bochner-Laplace), \etc,
are derived from $ \nabla $ in the usual way.  
In the example section \ref{sec:LdG} we use the surface Laplace operator $\DeltaS := (\Tr\nablaS^2):\tangentR[^2] \rightarrow \tangentR[^2]$ on 2-tensor fields,
which stated a kinda connection Laplace operator a priori \wrt\ surface derivative $ \nablaS $.
For more details see lemmas \ref{lem:laplace_equals_beltrami} and \ref{lem:laplace_equals_bochner},
where we show that  $ \DeltaS $ is the Laplace-Beltrami operator on the Cartesian Proxy components as well as a Bochner-like Laplace operator \wrt\ $ \nablaS $.
Additionally, corollary \ref{col:surface_laplace_decomposition} gives a relation to covariant differential operators  outside the Cartesian frame.
Inner products $ \inner{(\cdot)}{\cdot,\cdot} $ are written with angle brackets and labeled by its associated space.
For instance $ \inner{\tangentS[^2]}{\rb_1,\rb_2} = g_{ik}g_{jl}r_{1}^{ij}r_{2}^{kl}  $ is the local inner product,
or $ \innerH{\tangentS[^2]}{\rb_1,\rb_2} = \int_{\surf} \inner{\tangentS[^2]}{\rb_1,\rb_2} \dS $ is the global inner product 
of $ \rb_1,\rb_2\in\tangentS[^2] $. 
Note that inner products on tensor fields are backwards compatible with their subtensor fields, \eg\ it holds 
$ \inner{\tangentR[^2]}{\rb_1,\rb_2} = \inner{\tangentS[^2]}{\rb_1,\rb_2} $ for $ \rb_1,\rb_2\in\tangentS[^2] $.
Norms are given and written according to their inner products, \eg\ it is $ \normsq{\tangentS[^2]}{\rb} = \innerH{\tangentS[^2]}{\rb,\rb} $ valid for $ \rb\in\tangentS $.
We save writing an extra operation symbol, like a dot, for simple tensor-tensor multiplications $ \tangentR[^n]\times\tangentR[^m]\rightarrow\tangentR[^{n+m-2}] $,
\eg\ it is $ \Rb_1\Rb_2 = R_1^{AB} R_{2B} \eb_A\in\tangentR $ valid for $ \Rb_1\in\tangentR[^2] $ and $ \Rb_2\in\tangentR $.
However, we sometimes use the double-dot  symbol ``$ : $'' for the double-contraction product $  \tangentR[^n]\times\tangentR[^m]\rightarrow\tangentR[^{n+m-4}] $,
\eg\ it holds $ \Rb_1\operatorname{:}\Rb_2 = \inner{\tangentR[^2]}{\Rb_1,\Rb_2}\in\tangentS[^0]$ for $ \Rb_1,\Rb_2\in\tangentR[^2] $.
As in \cite{NitschkeSadikVoigt_A_2022}, we use arguments in square brackets to denote functional dependencies, 
\eg\ the scalar field 
$ f[\para_{\mfrak},\Vb_{\mfrak}] = \normsq{\R^3}{\para_{\mfrak}} + \normsq{\tangentR}{\Vb_{\mfrak}} \in\tangentS[^0] $ 
depends on the material surface parameterization $\para_{\mfrak}$ as a proxy for the surface, \wrt\ the material observer, as well as on its velocity $\Vb_{\mfrak}=\partial_t\para_{\mfrak}$.
Note that functional dependencies do not have to be mutual independent as we see in the former example.
In table \ref{tab:quantities_overview} frequently used quantities, also related to the chosen observer, are summarized.
We would also like to point out that \ref{sec:identities} contains a collection of lemmas, corollaries and their justifications that may be helpful 
for understanding the quantities in table \ref{tab:quantities_overview}.
For more details on observer related notations, see \cite{NitschkeVoigt_JoGaP_2022}.
Note that for the tangential material derivative, which is defined below, we use a dot over the field symbol.
A bar between a term and the dot parenthesize the term under the dot.
For instance, we write $ \dot{f}=\overdot{f_1  f_2} $ for $ f = f_1 f_2 $ in context of scalar fields. 

\subsection{Summary}\label{sec:summary}

In this section we provide a summary of the results in section \ref{sec:derivation} and relate them to the observer-invariant time derivatives derived in \cite{NitschkeVoigt_JoGaP_2022}.
Tables \ref{tab:vector_tangentialtimederivatives} and \ref{tab:ttensor_tangentialtimederivatives} give an overview about tangential time derivatives
on tangential vector fields in $ \tangentS $ and 2-tensor fields in $ \tangentS[^2] $, which are given in \cite{NitschkeVoigt_JoGaP_2022}.
We formulate these time derivatives \wrt\ an observer parameterization as well as a relation to the material time derivative.
Note that these time derivatives are special cases of instantaneous tensor fields in \cite{NitschkeVoigt_JoGaP_2022}.
We use the name prefix ``Jaumann'' synonymously to the prefix ``corotational''. 
Time derivatives on vector fields in $ \tangentR $ and 2-tensor fields in $ \tangentR[^2] $, which are derived in  section \ref{sec:derivation}, are summarized in 
table \ref{tab:vector_timederivatives} and \ref{tab:ttensor_timederivatives}.
We formulate these time derivatives in an orthogonal tangential-normal-decomposition for tangential-normal-decomposed tensor fields,
where we are able to use the corresponding tangential time-derivatives.
This representation could be useful for analytical perspectives.
In contrast, we give also a relation to the material derivative, which could be helpful for numerical implementations,
since it is possible to apply a Cartesian frame of the embedding space for the material derivative, 
\ie\ $ \Dmat\Rb=\dot{R}^A\eb_A $ for vector fields $ \Rb\in\tangentR $ and $ \Dmat\Rb=\dot{R}^{AB}\eb_A\otimes\eb_B $ for 2-tensor fields $ \Rb\in\tangentR $,
see \eqref{eq:tensor_matder_cartesian} for general $ n$-tensor fields.
A very useful property of the material and Jaumann derivative is their inner product compatibility,
\ie\ the inner product of tensor fields obey the product rule, see corollaries 
\ref{col:vector_innerprodcomp_Dmat}, \ref{col:vector_innerprodcomp_Djau} for vector fields and \ref{col:ttensor_innerprodcomp_Dmat}, \ref{col:ttensor_innerprodcomp_Djau}
for 2-tensor fields.
Likewise, they yield a compatible product rule with the tensor-vector product $ \tangentR[^2]\times\tangentR\rightarrow\tangentR $, 
see corollaries  \ref{col:ttensor_usualprodcomp_Dmat} and \ref{col:ttensor_usualprodcomp_Djau}.
Both convected derivatives do not exhibit these behaviors.
Note that the material derivative is not an extension of the tangential material derivative, contrary to all other time derivatives we present in this paper.
However, the pure tangential part of the material derivative yields such an extension,
which in turn describes an extension of \cite[Proposition 4]{NitschkeVoigt_JoGaP_2022} for non-tangential tensor fields.
\begin{table}[t]
\centering
\renewcommand{\arraystretch}{1.2}
\begin{tabular}{|c|c|c|c|}
\hline
\ldots derivative & identifier & \wrt\ observer coordinates & \wrt\ $ \dot{\rb} $\\
\hline
material 
    & $ \dot{\rb} $ 
        & $ (\partial_t r^i)\partial_i\para_{\ofrak} + \nabla_{\ub}\rb + \Gb[\Vb_{\ofrak}]\rb $ 
            & $ \dot{\rb} $ \\
\hline
upper-convected 
    & $ \timeLu\rb $
        & $ (\partial_t r^i)\partial_i\para_{\ofrak} + \nabla_{\ub}\rb - \nabla_{\rb}\ub $
            &  $ \dot{\rb} - \Gb[\Vb_{\mfrak}]\rb$\\
\hline
lower-convected 
    & $ \timeLl\rb $
        & $ g_{\ofrak}^{ij}(\partial_t r_j)\partial_i\para_{\ofrak} + \nabla_{\ub}\rb + \rb\nabla\ub $
            &  $ \dot{\rb} + \Gb^T[\Vb_{\mfrak}]\rb$\\
\hline
Jaumann
    & $ \timeJ\rb $
        &  $\frac{1}{2}(\timeLu\rb + \timeLl\rb)$
            &  \splitcelltab{$ \dot{\rb} -\Ab[\Vb_{\mfrak}]\rb$\\
                                $ = \dot{\rb} - \frac{1}{2} (\rot\vb_{\mfrak}) *\rb $}\\
\hline
\end{tabular}
\caption{Tangential time derivatives on tangential vector fields $ \rb = r^i\partial_i\para_{\ofrak} \in\tangentS $, taken from \cite{NitschkeVoigt_JoGaP_2022}.}
\label{tab:vector_tangentialtimederivatives}

\bigskip

\begin{tabular}{|c|c|c|c|}
\hline
\ldots derivative & identifier & \wrt\ tangential derivatives & \wrt\ $ \Dmat\Rb $\\
\hline
material 
    & $ \Dmat\Rb $ 
        &  $\dot{\rb} - \phi \bb[\Vb_{\mfrak}] + \left( \dot{\phi} + \inner{\tangentS}{\rb , \bb[\Vb_{\mfrak}]} \right) \normal$
            & $\Dmat\Rb $ \\
\hline
upper-convected 
    & $ \Dupp\Rb $
        & $ \timeLu\rb + \dot{\phi}\normal$
            & $ \Dmat\Rb - \Gbcal[\Vb_{\mfrak}]\Rb  $ \\
\hline
lower-convected 
    & $ \Dlow\Rb $
        & $ \timeLl\rb + \dot{\phi}\normal$
                    & $ \Dmat\Rb + \Gbcal^T[\Vb_{\mfrak}]\Rb  $  \\
\hline
Jaumann
    & $ \Djau\Rb $
        &  $\timeJ\rb + \dot{\phi}\normal$
            &  $ \Dmat\Rb - \Abcal[\Vb_{\mfrak}]\Rb  $\\
\hline
\end{tabular}
\caption{Time derivatives on vector fields $ \Rb = \rb + \phi\normal \in \tangentR $, where $ \rb\in\tangentS $ and $ \phi\in\tangentS[^0] $.}
\label{tab:vector_timederivatives}
\end{table}

\begin{table}[t]
\centering
\renewcommand{\arraystretch}{1.2}
\begin{tabular}{|c|c|c|c|}
\hline
\ldots derivative & identifier & \wrt\ observer coordinates & \wrt\ $ \dot{\rb} $\\
\hline
material 
    & $ \dot{\rb} $ 
        & \splitcelltab{$ (\partial_t r^{ij})\partial_i\para_{\ofrak}\otimes\partial_j\para_{\ofrak} + \nabla_{\ub}\rb $\\
                        $+ \Gb[\Vb_{\ofrak}]\rb + \rb\Gb^T[\Vb_{\ofrak}] $ }
            & $ \dot{\rb} $ \\
\hline
upper-convected 
    & $ \timeLuu\rb $
        & \splitcelltab{$ (\partial_t r^{ij})\partial_i\para_{\ofrak}\otimes\partial_j\para_{\ofrak} + \nabla_{\ub}\rb $\\
                        $- (\nabla\ub)\rb - \rb(\nabla\ub)^T  $ }
            &  $ \dot{\rb} - \Gb[\Vb_{\mfrak}]\rb - \rb\Gb^T[\Vb_{\mfrak}]$\\
\hline
lower-convected 
    & $ \timeLll\rb $
        & \splitcelltab{$ g_{\ofrak}^{ik}g_{\ofrak}^{jl}(\partial_t r_{kl})\partial_i\para_{\ofrak}\otimes\partial_j\para_{\ofrak} + \nabla_{\ub}\rb $\\
                        $+ (\nabla\ub)^T\rb + \rb(\nabla\ub)  $ }
            &  $ \dot{\rb} + \Gb^T[\Vb_{\mfrak}]\rb + \rb\Gb[\Vb_{\mfrak}]$\\
\hline
Jaumann
    & $ \timeJ\rb $
        &  $\frac{1}{2}(\timeLuu\rb + \timeLll\rb)$
            &  \splitcelltab{$ \dot{\rb} - \Ab[\Vb_{\mfrak}]\rb + \rb\Ab[\Vb_{\mfrak}]$\\
                                $ = \dot{\rb} - 2\Ab[\Vb_{\mfrak}]\projQS\rb$\\
                                $ = \dot{\rb} - \frac{1}{2} (\rot\vb_{\mfrak})(*_1\rb + *_2\rb)$}\\
\hline
\end{tabular}
\caption{Tangential time derivatives on tangential 2-tensor fields $ \rb = r^{ij}\partial_i\para_{\ofrak}\otimes\partial_j\para_{\ofrak} \in\tangentS[^2] $, taken from \cite{NitschkeVoigt_JoGaP_2022}.}
\label{tab:ttensor_tangentialtimederivatives}

\bigskip

\begin{tabular}{|c|c|c|c|}
\hline
\ldots derivative & identifier & \wrt\ tangential derivatives & \wrt\ $ \Dmat\Rb $\\
\hline
material 
    & $ \Dmat\Rb $ 
        &  \splitcelltab{$\dot{\rb} - \etab_L\otimes\bb[\Vb_{\mfrak}] - \bb[\Vb_{\mfrak}]\otimes\etab_R$\\
                        $ +\left( \dot{\etab}_{L} + \rb\bb[\Vb_{\mfrak}] - \phi\bb[\Vb_{\mfrak}] \right)\otimes\normal$\\
                        $ +\normal\otimes\left( \dot{\etab}_{R} + \bb[\Vb_{\mfrak}]\rb - \phi\bb[\Vb_{\mfrak}] \right) $\\
                        $ + \left( \dot{\phi} + \inner{\tangentS}{\etab_{L}+\etab_{R}, \bb[\Vb_{\mfrak}]} \right)\normal\otimes\normal $}
            & $\Dmat\Rb $ \\
\hline
upper-convected 
    & $ \Dupp\Rb $
        & $ \timeLuu\rb + \timeLu\etab_{L}\otimes\normal + \normal\otimes\timeLu\etab_{R} +  \dot{\phi}\normal\otimes\normal $
            & $ \Dmat\Rb - \Gbcal[\Vb_{\mfrak}]\Rb - \Rb\Gbcal^T[\Vb_{\mfrak}] $ \\
\hline
lower-convected 
    & $ \Dlow\Rb $
        & $ \timeLll\rb + \timeLl\etab_{L}\otimes\normal + \normal\otimes\timeLl\etab_{R} +  \dot{\phi}\normal\otimes\normal $
                    & $ \Dmat\Rb + \Gbcal^T[\Vb_{\mfrak}]\Rb + \Rb\Gbcal[\Vb_{\mfrak}] $  \\
\hline
Jaumann
    & $ \Djau\Rb $
        &  $\timeJ\rb + \timeJ\etab_{L}\otimes\normal + \normal\otimes\timeJ\etab_{R} +  \dot{\phi}\normal\otimes\normal$
            &  $ \Dmat\Rb - \Abcal[\Vb_{\mfrak}]\Rb + \Rb\Abcal[\Vb_{\mfrak}] $\\
\hline
\end{tabular}
\caption{Time derivatives on 2-tensor fields $ \Rb = \rb + \etab_L\otimes\normal + \normal\otimes\etab_R + \phi\normal\otimes\normal \in \tangentR[^2]$
where $ \rb\in\tangentS[^2] $, $ \etab_L,\etab_R \in\tangentS $ and $ \phi\in\tangentS[^0] $.}
\label{tab:ttensor_timederivatives}
\end{table}

In context of 2-tensor fields we consider Q-tensor fields $ \tangentQR < \tangentR[^2] $ as a subbundle,
where our attention is mainly directed to the material and Jaumann derivative.
Since $ \tangentQR $ is closed \wrt\ both derivatives, we could use the more general representations for $ \tangentR[^2] $ in table \ref{tab:ttensor_timederivatives}.
These time derivatives apply in the surface Landau-de Gennes model \eqref{eq:qtensor_LdG_flow}.
A more aligned formulation of the material and Jaumann derivative, \wrt\ the orthogonal decomposition \eqref{eq:qtensor_decomposition} for Q-tensor fields, can be found
in \ref{eq:qtensor_Dmat} and \ref{eq:qtensor_Djau}.
We also consider surface conforming Q-tensor fields $ \tangentCQR < \tangentQR $, which are not closed by the material derivative but by the Jaumann derivative.
Hence we present an adjusted material derivative $ \DCQmat:=\projCQR\circ\Dmat $ under the aid of the unique orthogonal
projection $ \projCQR: \tangentQR \rightarrow \tangentCQR   $.
This surface conforming material derivative and the Jaumann derivative are used in the surface conforming Landau-de Gennes model \eqref{eq:qtensor_LdG_confomalflow}.
An orthogonal decomposition of both time derivatives on surface conforming Q-tensor fields can be found in table \ref{tab:cqtensor_timederivatives}
and apply in the equivalent formulation \eqref{eq:qtensor_LdG_confomalflow_decomposed} of the  surface conforming Landau-de Gennes model.
\begin{table}[t]
\centering
\renewcommand{\arraystretch}{1.2}
\begin{tabular}{|c|c|c|}
\hline
\ldots Derivative & Identifier & \wrt\ tangential derivatives \\
\hline
Surface conforming Material 
    & $ \DCQmat\Qb $ 
        &  $\dot{\qb} + \dot{\beta} \left( \normal\otimes\normal - \frac{1}{2}\IdS \right)$ \\
\hline
Jaumann
    & $ \Djau\Qb $
        &  $\timeJ\qb + \dot{\beta} \left( \normal\otimes\normal - \frac{1}{2}\IdS \right)$\\
\hline
\end{tabular}
\caption{Time derivatives on surface conforming Q-tensor fields 
$ \Qb = \qb + \beta (\normal\otimes\normal - \frac{1}{2}\IdS)\in\tangentCQR<\tangentQR < \tangentR[^2] $,
where $ \qb\in\tangentQS $ and $ \beta\in\tangentS[^0] $.}
\label{tab:cqtensor_timederivatives}
\end{table}

\section{Derivations}\label{sec:derivation}

\subsection{General Approach and Scalar Fields}\label{sec:general_approach}

Formally, we could define an arbitrary time-derivative on $ \Rb\in\tangentR[^n] $ by
\begin{align}
    (\Dt\Rb)[\para_{\mfrak}](t,y_{\mfrak}^1,y_{\mfrak}^2) \label{eq:tensor_timeD_general}
            &:= \lim_{\tau\rightarrow 0} \frac{1}{\tau}\left( (\Phi^*_{t,\tau}\Rb[\para_{\mfrak}]\vert_{t+\tau})(t,y_{\mfrak}^1,y_{\mfrak}^2) 
                                                              - \Rb[\para_{\mfrak}](t,y_{\mfrak}^1,y_{\mfrak}^2) \right)
\end{align}
where 
$ \Phi^*_{t,\tau}:\tangent^n\R^3\vert_{\surf,t+\tau} \rightarrow \tangent^n\R^3\vert_{\surf,t} $ 
is a convenient pullback by the map
\begin{align*}
    \Phi_{t,\tau}:\surf\vert_{t} \rightarrow \surf\vert_{t+\tau}: 
        \quad\para_{\mfrak}(t,y_{\mfrak}^1,y_{\mfrak}^2) \mapsto \para_{\mfrak}(t+\tau,y_{\mfrak}^1,y_{\mfrak}^2) \formPeriod
\end{align*}
Even if  the time derivative is described by a material observer, \wrt\ its parameterization $ \para_{\mfrak} $, 
we are able to evaluate \eqref{eq:tensor_timeD_general} by an arbitrary observer, \wrt\ parameterization $ \para_{\ofrak} $ with the aid of
relation
\begin{align}\label{eq:tensor_observer_relation}
     \Rb[\para_{\mfrak}](t,y_{\mfrak}^1,y_{\mfrak}^2) 
        = \Rb[\para_{\ofrak}](t,(\para_{\ofrak}\vert_{t}^{-1}\circ\para_{\mfrak})(t,y_{\mfrak}^1,y_{\mfrak}^2))
            \in\tangent^n_{\para_{\mfrak}(t,y_{\mfrak}^1,y_{\mfrak}^2)}\R^3\vert_{\surf} \formComma
\end{align}
respectively, the inverse relation
\begin{align*}
    \Rb[\para_{\ofrak}](t, y_{\ofrak}^1,y_{\ofrak}^2)
        = \Rb[\para_{\mfrak}](t, (\para_{\mfrak}\vert_{t}^{-1} \circ \para_{\ofrak} )(t, y_{\ofrak}^1,y_{\ofrak}^2))
            \in\tangent^n_{\para_{\ofrak}(t,y_{\ofrak}^1,y_{\ofrak}^2)}\R^3\vert_{\surf}\formPeriod
\end{align*}
The general proceeding is to assume a pullback, conclude the associated time derivative \wrt\ the material observer and transform it \wrt\ an arbitrary observer to establish observer-invariance.

For scalar fields $ f\in\tangentR[^0]=\tangentS[^0] $, \ie\ $ n=0 $, the only noteworthy pullback is simply given by
\begin{align} \label{eq:scalar_pullback}
    (\Phi^{*_0}_{t,\tau}f[\para_{\mfrak}]\vert_{t+\tau})(t,y_{\mfrak}^1,y_{\mfrak}^2) 
        &=  f[\para_{\mfrak}](t+\tau,y_{\mfrak}^1,y_{\mfrak}^2)
           \in\tangent^0_{\para_{\mfrak}(t,y_{\mfrak}^1,y_{\mfrak}^2)}\surf\formPeriod
\end{align}  
Hence, with $ \dot{f}:= \Dt\vert_{\Phi^{*}_{t,\tau}=\Phi^{*_0}_{t,\tau}}f $, \eqref{eq:tensor_timeD_general} becomes
\begin{align*}
    \dot{f}[\para_{\mfrak}](t,y_{\mfrak}^1,y_{\mfrak}^2)
        &= \partial_t f  [\para_{\mfrak}](t,y_{\mfrak}^1,y_{\mfrak}^2) 
           \in\tangent^0_{\para_{\mfrak}(t,y_{\mfrak}^1,y_{\mfrak}^2)}\surf
\end{align*}
for a material observer. 
Getting the time derivative $ \dot{f}[\para_{\ofrak}] $ for an arbitrary observer given by parameterization $ \para_{\ofrak} $ is more difficult.
The time derivative \eqref{eq:tensor_timeD_general} as well as the pullback \eqref{eq:scalar_pullback} have to be evaluated \wrt\ the relation \eqref{eq:tensor_observer_relation}.
As we can see in appendix \ref{sec:scalar_timeder_scratch}, applying a Taylor expansion to this pullback at $ \tau=0 $ leads to 
\begin{align*}
    \dot{f}[\para_{\ofrak}](t,y_{\ofrak}^1,y_{\ofrak}^2)  
            &= \partial_t f[\para_{\ofrak}](t,y_{\ofrak}^1,y_{\ofrak}^2)  + (\nabla_{\ub} f)[\para_{\ofrak}](t,y_{\ofrak}^1,y_{\ofrak}^2) 
            \in\tangent^0_{\para_{\ofrak}(t,y_{\ofrak}^1,y_{\ofrak}^2)}\surf \formComma\\
   \text{where}\quad \ub
            =\ub[\para_\ofrak,\para_\mfrak](t,y^1_\ofrak,y^1_\ofrak)  
           &:= \Vb_{\mfrak}[\para_\mfrak](t,(\para_{\mfrak}\vert_{t}^{-1} \circ \para_{\ofrak} )(t, y_{\ofrak}^1,y_{\ofrak}^2)) - \Vb_\ofrak[\para_\ofrak](t,y^1_\ofrak,y^1_\ofrak)
\end{align*}
is the relative velocity, 
$\Vb_\ofrak[\para_\ofrak](t,y^1_\ofrak,y^1_\ofrak):= \partial_t \para_{\ofrak}(t,y^1_\ofrak,y^1_\ofrak)$ the observer velocity
and $\Vb_\mfrak[\para_\mfrak](t,y^1_\mfrak,y^1_\mfrak):= \partial_t \para_{\ofrak}(t,y^1_\mfrak,y^1_\mfrak)$ the material velocity.
This is also consistent with the scalar-valued time derivative given in   \cite{NitschkeVoigt_JoGaP_2022}.
Since the observer is arbitrary and for sake of simplicity, we also write 
\begin{align} \label{eq:scalar_timeder}
    \dot{f} &= \partial_t f + \nabla_{\ub}f \in\tangentS[^0]
\end{align}
for short, which is the common form in context of non(-Einstein)-relativistic settings \cite{Van_2008}, 
and ALE (Arbitrary Lagrangian–Eulerian) methods on non-stationary surfaces \cite{Sahu_2020, Ramaswamy_1987}.
A material perspective, \ie\ $ \ub=0 $, applies to Lagrangian particle methods \cite{Idelsohn_2004,Koh_2011} for instance.
If $f$ is extended in a volume around $\surf$, there are also alternative formulation of $\dot{f}$,
see  \cite{Dziuk2007} for instance.

Since we consider $ \R^3 $ quantities, even though restricted to the surface, we show at the end of each of the following subsections that all
considered time derivatives are consistent to their counterpart in a volume, 
\ie\ the thin film limit of a time derivative in a bulk equals its time derivatives on the surface.
We use the thin film parameterization $ \chib[\para] $, defined by
\begin{align}\label{eq:tf_para}
    \chib[\para](t,y^1,y^2,\xi) 
        &:= \para(t,y^1,y^2) + \xi\normal[\para](t,y^1,y^2)
\end{align}
with $ \xi\in[-h,h] $,
to describe the thin film $ \surf_h $ around $ \surf $, see \cite{Nitschke_2018} for more details.
Therefore, $\chib[\para_{\mfrak}]$ is the material and $ \chib[\para_{\ofrak}] $ an arbitrary observer thin film parameterization.
According to this, $ \widehat{\Vb}_{\mfrak} := \partial_t\chib[\para_{\mfrak}] $ is the material and 
$ \widehat{\Vb}_{\ofrak} := \partial_t\chib[\para_{\ofrak}] $ the observer thin film velocity.
We obtain the relative thin film velocity
\begin{align*}
    \widehat{\Vb}_{\mfrak} - \widehat{\Vb}_{\ofrak}
        &= \ub - \xi\shop\ub
\end{align*}
as a consequence by \eqref{eq:partialtimeder_normal}.
For extended scalar fields $ \widehat{f}\in\tangent[^0]{\surf_h} $, which are sufficing $ \widehat{f}\vert_{\xi=0} = f\in\tangentS[^0] $,
we use the Taylor expansion
\begin{align*}
    \widehat{f} 
        &= f + \xi(\partial_{\xi}\widehat{f})\vert_{\xi=0} + \landau(\xi^2)
\end{align*}
at $ \xi=0 $.
Note that the normal coordinate $ \xi $ and the time parameter $ t $ are mutually independent,
\ie\  $ \partial_{\xi}$ and $\partial_t  $ are commuting on scalar fields.
Eventually, this yields
\begin{align}\label{eq:tfl_scalar_timeder}
    \dot{\widehat{f}}
        &= \partial_t \widehat{f} + \nablahat_{\widehat{\Vb}_{\mfrak} - \widehat{\Vb}_{\ofrak}} \widehat{f}
         = \dot{f} + \landau(\xi)\formComma
\end{align}
\ie\ it holds $ \dot{\widehat{f}} \rightarrow \dot{f} $ for $ h\rightarrow 0 $.

\subsection{Material Derivative} \label{sec:Dmat}

In order to obtain the material time derivative we could simply use the Cartesian frame $ \{\eb_A\} $, which is Eulerian and constant in space.
Though it seems that an additional frame, which is not given by the chart through the parameterization, would complicate the situation at first glance, 
the material pullback becomes quite easy. 
The pullback implements the scalar pullback \eqref{eq:scalar_pullback} on each Cartesian component.
This yields the definition
\begin{align*}
    (\Phi^{*_m}_{t,\tau}\Rb[\para_{\mfrak}]\vert_{t+\tau})(t,y_{\mfrak}^1,y_{\mfrak}^2) 
        &:=  R^{A_1 \ldots A_n}[\para_{\mfrak}](t+\tau,y_{\mfrak}^1,y_{\mfrak}^2) \bigotimes_{\alpha=1}^n \eb_{A_\alpha} 
           \in\tangent^n_{\para_{\mfrak}(t,y_{\mfrak}^1,y_{\mfrak}^2)}\R^3\vert_{\surf}\formPeriod
\end{align*}  
Therefore the material derivative is given by
$ \Dmat := \Dt\vert_{\Phi^{*}_{t,\tau}=\Phi^{*_m}_{t,\tau}} $, \ie
\begin{align*}
    (\Dmat\Rb)[\para_{\mfrak}](t,y_{\mfrak}^1,y_{\mfrak}^2)  
        &= \partial_t R^{A_1 \ldots A_n}[\para_{\mfrak}](t,y_{\mfrak}^1,y_{\mfrak}^2)\bigotimes_{\alpha=1}^n \eb_{A_\alpha} 
        \in\tangent^n_{\para_{\mfrak}(t,y_{\mfrak}^1,y_{\mfrak}^2)}\R^3\vert_{\surf} \formComma
\end{align*}  
for the material observer.
Since the frame is constant, we only have to consider the scalar Cartesian proxy fields $ R^{A_1 \ldots A_n}[\para_{\mfrak}]\in\tangentS[^0] $.
For an arbitrary observer, \eqref{eq:scalar_timeder} yields 
\begin{align}\label{eq:tensor_matder_cartesian}
     \Dmat\Rb = \dot{R}^{A_1 \ldots A_n} \bigotimes_{\alpha=1}^n \eb_{A_\alpha} \formPeriod
\end{align}
One first observation of \eqref{eq:tensor_matder_cartesian} is that this time derivative equals the material time-derivative in a volume up to the restriction to the surface, 
\ie\ it does not depend on behaviors of the surface at all.
%\footnote{That is not a big surprise, since from a physical point of view, a quantity bounded to a single material particle only by its location is transported in a force-free situation utterly independent of the surrounding space, \ie\ a material derivative should not depend on the surrounding setting and can be calculated locally at a single material point in motion.}
This is not to be expected by other time-derivatives.
Moreover, \eqref{eq:tensor_matder_cartesian}, contrary to \eqref{eq:tensor_timeD_general} in general, is now represented in context of an arbitrary observer chart,
\ie\ all we have to do is calculating \eqref{eq:tensor_matder_cartesian} also in terms of an arbitrary extended surface observer frame 
$ \{\partial_1\para_{\ofrak}, \partial_2\para_{\ofrak}, \normal\} $.
Note that the Cartesian frame
yields
\begin{align*} %\label{eq:cartesian2observerframe}
    \eb_A
        &= \delta_{AB}\left( g_{\ofrak}^{ij}\partial_j\paraC_{\ofrak}^B\partial_i\para_\ofrak + \normalC^{B}\normal \right)\formComma
\end{align*}
at all local events $ (t,y_{\ofrak}^1,y_{\ofrak}^2) $.
In the following subsections we transform the frame and the associated proxy fields to the extended observer frame especially for vector and 2-tensor fields. 

For extended tensor fields $ \widehat{\Rb}\in\tangent[^n]{\surf_h} $, which are sufficing $ \widehat{\Rb}\vert_{\xi=0} = \Rb\in\tangentR[^n] $,
we conclude from \eqref{eq:tensor_matder_cartesian} and \eqref{eq:tfl_scalar_timeder} that
\begin{align} \label{eq:tfl_Dmat}
    \dot{\widehat{\Rb}}
        &=  \dot{\widehat{R}}^{A_1 \ldots A_n} \bigotimes_{\alpha=1}^n \eb_{A_\alpha}
        \rightarrow \Dmat\Rb
\end{align}
is valid for $ h \rightarrow 0 $.

\subsubsection{Vector Fields}

To represent the material derivative $ \Dmat\Rb = \dot{R}^A \eb_A $ \eqref{eq:tensor_matder_cartesian} on vector fields $ \Rb\in\tangentR $, 
we use the orthogonal decomposition
\begin{align} \label{eq:vector_decomposition}
    \Rb &= \rb + \phi\normal \in \tangentR\formComma
\end{align}
where $ \rb\in\tangentS $ and $ \phi\in\tangentS[^0] $ are given by $ \Rb $ uniquely.
The tangential covariant observer proxy of $ \Dmat\Rb $ yields
\begin{align*}
    \inner{\tangentR}{ \Dmat\Rb, \partial_k\para_{\ofrak} }
        &= \delta_{AB} \dot{R}^A \partial_k\para_{\ofrak}^B
         = \overdotinner{\tangentR}{\Rb, \partial_k\para_{\ofrak}}
            - R_B \overdot{\partial_k\paraC_{\ofrak}^B}\\
        &= \partial_t r_k + u^i\partial_i r_k - \left( r_B + \phi \normalC_B \right) \left( \partial_k V_{\ofrak}^B +  u^i\partial_i\partial_k\paraC_{\ofrak}^B \right)
\end{align*}
by time derivative \eqref{eq:scalar_timeder} on scalar fields and decomposition \eqref{eq:vector_decomposition}.
With \eqref{eq:vector_partialtimeder_covar}, \eqref{eq:partialder_V} and \eqref{eq:partialpartialder_X} we obtain
\begin{align*}
    \inner{\tangentR}{ \Dmat\Rb, \partial_k\para_{\ofrak} }
        &= g_{\ofrak kj} \partial_t r^j + u^i\left( \partial_i r_k - \Gamma_{\ofrak ik}^j r_j \right)
            + G_{ij}[\Vb_{\ofrak}]r^j - \phi\left( u^i\shopC_{ik} + b_k[\Vb_{\ofrak}] \right) \\
        &=  g_{\ofrak kj} \partial_t r^j + u^i r_{k|i} + G_{ij}[\Vb_{\ofrak}]r^j -\phi b_k[\Vb_{\mfrak}]
        = \left[ \dot{\rb} - \phi \bb[\Vb_{\mfrak}] \right]_k \formComma
\end{align*}
where $ \dot{\rb}\in\tangentS $ is the material derivative of the tangential vector field $ \rb $ given in table \ref{tab:vector_tangentialtimederivatives}.
For the normal part of $ \Dmat\Rb $ we use the time derivative \eqref{eq:scalar_timeder} on scalar fields again and the rate of the normal field given in \eqref{eq:timederwithu_normal}.
This yields
\begin{align*}
    \inner{\tangentR}{ \Dmat\Rb, \normal }
        &= \delta_{AB} \dot{R}^{A} \normalC^B
         = \overdotinner{\tangentR}{\Rb, \normal} - R_{B}\dot{\normalC}^{B}
         = \dot{\phi} + \inner{\tangentS}{\rb , \bb[\Vb_{\mfrak}]} \formPeriod
\end{align*}
\begin{corollary}
    For all $ \Rb = \rb + \phi\normal \in \tangentR $, $ \rb\in\tangentS $ and $ \phi\in\tangentS[^0] $ holds
    \begin{align} \label{eq:vector_Dmat}
        \Dmat\Rb 
            &= \dot{\rb} - \phi \bb[\Vb_{\mfrak}] + \left( \dot{\phi} + \inner{\tangentS}{\rb , \bb[\Vb_{\mfrak}]} \right) \normal \formPeriod
    \end{align}
\end{corollary}
Note that $\Dmat\Vb_{\mfrak} $ equals the material acceleration in an observer-invariant representation, see \cite{NitschkeVoigt_JoGaP_2022,YavariOzakinSadik_JoNS_2016}. 
To show inner product compatibility of the material derivative, we use that the proxy $ \delta_{AB} $ of the Cartesian metric tensor is in the kernel
of the scalar time derivative \eqref{eq:scalar_timeder}, 
\ie\ it holds $ \dot{\delta}_{AB} = \partial_t\delta_{AB} + u^k\partial_k\delta_{AB} = 0 $.
Hence, we obtain $ \overdotinner{\tangentR}{\Rb_1, \Rb_2} = \delta_{AB}( \dot{R}_1^A R_2^B + R_1^A \dot{R}_2^B ) $ 
for all $ \Rb_1,\Rb_2\in\tangentR $, which gives the following corollary.
\begin{corollary}\label{col:vector_innerprodcomp_Dmat}
    The material derivative on vector fields is compatible with the inner product,
    \ie\ for all $ \Rb_1 = \rb_1 + \phi_1\normal,\Rb_2 = \rb_2 + \phi_2\normal\in\tangentR $ holds
    \begin{align}
        \overdotinner{\tangentR}{\Rb_1, \Rb_2} \label{eq:vector_inner_Dmat}
            &= \inner{\tangentR}{ \Dmat\Rb_1 , \Rb_2 } + \inner{\tangentR}{ \Rb_1 , \Dmat\Rb_2 }\\
            &= \inner{\tangentS}{ \dot{\rb}_1 , \rb_2} + \inner{\tangentS}{ \rb_1 , \dot{\rb}_2} 
                + \dot{\phi}_1 \phi_2 + \phi_1 \dot{\phi}_2 \formPeriod \notag
    \end{align}
\end{corollary}

\subsubsection{2-Tensor Fields}
To represent the material derivative \eqref{eq:tensor_matder_cartesian} on 2-tensor fields $ \Rb $, 
\ie\ $ \Dmat\Rb = \dot{R}^{AB} \eb_A\otimes\eb_B $, we use the orthogonal decomposition
\begin{align} \label{eq:ttensor_decomposition}
    \Rb &= \rb + \etab_L\otimes\normal + \normal\otimes\etab_R + \phi\normal\otimes\normal \in \tangentR[^2]\formComma
\end{align}
where $ \rb\in\tangentS[^2] $, $ \etab_L,\etab_R \in\tangentS $ and $ \phi\in\tangentS[^0] $ are given by $ \Rb $ uniquely.
The tangential covariant observer proxy of $ \Dmat\Rb $ yields
\begin{align*}
    \MoveEqLeft \inner{\tangentR[^2]}{\Dmat\Rb, \partial_m \para_{\ofrak} \otimes \partial_n\para_{\ofrak}}\\
     &= \dot{R}^{AB}\delta_{AC}\delta_{BD} \partial_m\paraC_{\ofrak}^C \partial_n\paraC_{\ofrak}^D\\
     &= \overdotinner{\tangentR[^2]}{\Rb, \partial_m \para_{\ofrak} \otimes \partial_n\para_{\ofrak}}
                  - R_{CD} \left( \overdot{\partial_m\paraC_{\ofrak}^C}\partial_n\paraC_{\ofrak}^D 
                  + \partial_m\paraC_{\ofrak}^C \overdot{\partial_n\paraC_{\ofrak}^D}\right)\\
     &= \partial_t r_{mn} + u^k \partial_k r_{mn}
        - R_{CD}
           \left( \partial_m V_{\ofrak}^C \partial_n\paraC_{\ofrak}^D + \partial_m\paraC_{\ofrak}^C \partial_n V_{\ofrak}^D 
                  +  u^k\partial_k\partial_m\paraC_{\ofrak}^C\partial_n \paraC_{\ofrak}^D  
                  + u^k \partial_m\paraC_{\ofrak}^C \partial_k\partial_n\paraC_{\ofrak}^D \right)
\end{align*}
by time derivative \eqref{eq:scalar_timeder} on scalar fields and decomposition \eqref{eq:ttensor_decomposition},
which is read $ R_{CD} = r_{CD} + \eta_{LC}\normalC_{D} + \normalC_{C}\eta_{RD} + \phi\normalC_{C}\normalC_{D} $
in the Cartesian proxy notation.
With \eqref{eq:ttensor_partialtimeder_covar}, \eqref{eq:partialder_V} and \eqref{eq:partialpartialder_X} we obtain
\begin{align*}
    \MoveEqLeft \inner{\tangentR[^2]}{\Dmat\Rb, \partial_m \para_{\ofrak} \otimes \partial_n\para_{\ofrak}} \\
    &= g_{\ofrak mi} g_{\ofrak nj} \partial_t r^{ij} + u^k \partial_k r_{mn}
            + r_{in} \tensor{G}{_m^i}[\Vb_{\ofrak}] + r_{mi}\tensor{G}{_n^i}[\Vb_{\ofrak}]
            -\eta_{R n} b_{m}[\Vb_{\ofrak}] - \eta_{L m} b_{n}[\Vb_{\ofrak}] \\
      &\quad\quad -u^k\left( \tensor{r}{^i_n}\Gamma_{\ofrak mki} + \tensor{r}{_m^i}\Gamma_{\ofrak nki} 
                            +\eta_{R n}\shopC_{mk} + \eta_{L m}\shopC_{nk} \right)\\
    &= g_{\ofrak mi} g_{\ofrak nj} \partial_t r^{ij} + u^k r_{mn|k}
       + r_{in} \tensor{G}{_m^i}[\Vb_{\ofrak}] + r_{mi}\tensor{G}{_n^i}[\Vb_{\ofrak}]
       - \eta_{R n} b_{m}[\Vb_{\mfrak}] - \eta_{L m} b_{n}[\Vb_{\mfrak}]\\
    &= \left[ \dot{\rb} - \etab_L\otimes\bb[\Vb_{\mfrak}] - \bb[\Vb_{\mfrak}]\otimes\etab_R  \right]_{mn}\formComma
\end{align*}
where $ \dot{\rb}\in\tangentS[^2] $ is the material derivative of the tangential 2-tensor field $ \rb $ given in table \ref{tab:ttensor_tangentialtimederivatives}.
In the same manner we calculate the covariant observer proxy of the tangential-normal part.
Hence, with \eqref{eq:scalar_timeder}, \eqref{eq:ttensor_decomposition}, \eqref{eq:vector_partialtimeder_covar}, \eqref{eq:partialder_V}, \eqref{eq:timederwithu_normal} and \eqref{eq:partialpartialder_X},
we get
\begin{align*}
    \MoveEqLeft \inner{\tangentR[^2]}{\Dmat\Rb, \partial_m \para_{\ofrak} \otimes \normal} \\
    &= \dot{R}^{AB} \delta_{AC}\delta_{BD}  \partial_m\paraC_{\ofrak}^C \normalC^D
     = \overdotinner{\tangentR[^2]}{\Rb, \partial_m \para_{\ofrak} \otimes \normal}
             - R_{CD} \left( \overdot{\partial_m\paraC_{\ofrak}^C}\normalC^D 
                             + \partial_m\paraC_{\ofrak}^C \dot{\normalC}^D\right)\\
    &= \partial_t \eta_{L m} + u^k\partial_k\eta_{L m}
        -R_{CD}\left( \partial_m V_{\ofrak}^C \normalC^D + u^k \partial_k \partial_m X_{\ofrak}^C \normalC^D - \partial_m\paraC_{\ofrak}^C b^D[\Vb_{\mfrak}]   \right)\\
     &= g_{\ofrak mi} \partial_t \eta_L^i + u^k\partial_k\eta_{L m}  - u^k \eta_L^i \Gamma_{\ofrak kmi}
                +  G_{mi}[\Vb_{\ofrak}]\eta_L^i - \phi b_{m}[\Vb_{\ofrak}]  -  \phi u^k \shopC_{mk}
                + r_{mi} b^i[\Vb_{\mfrak}]\\
     &= g_{\ofrak mi} \partial_t \eta_L^i + u^k \eta_{L m|k} +  G_{mi}[\Vb_{\ofrak}]\eta_L^i - \phi b_{m}[\Vb_{\mfrak}] + r_{mi} b^i[\Vb_{\mfrak}]
     = \left[ \dot{\etab}_{L} + \rb\bb[\Vb_{\mfrak}] - \phi\bb[\Vb_{\mfrak}] \right]_m \formComma
\end{align*}
where $ \dot{\etab}_{L}\in\tangentS $ is the material derivative of the tangential vector field $ \etab_{L} $ given in table \ref{tab:ttensor_tangentialtimederivatives}.
Since the material derivative is compatible with transposition, \ie\ it is $ \Dmat\Rb^T := \Dmat(\Rb^T) = (\Dmat\Rb)^T  $ valid,
we get the normal-tangential part by
\begin{align*}
    \inner{\tangentR[^2]}{\Dmat\Rb, \normal \otimes \partial_n \para_{\ofrak} }
        &= \inner{\tangentR[^2]}{\Dmat\Rb^T, \partial_n \para_{\ofrak} \otimes \normal}
        = \left[ \dot{\etab}_{R} + \bb[\Vb_{\mfrak}]\rb - \phi\bb[\Vb_{\mfrak}] \right]_n
\end{align*} 
as a consequence.
The pure normal part of the material derivative yields
\begin{align*}
    \inner{\tangentR[^2]}{\Dmat\Rb, \normal\otimes\normal}
        &= \dot{R}^{AB} \delta_{AC}\delta_{BD} \normalC^{C} \normalC^{D}
        =  \overdotinner{\tangentR[^2]}{\Rb, \normal \otimes \normal}  
                - R_{CD} \left(\dot{\normalC}^C\normalC^D 
               + \normalC^C \dot{\normalC}^D\right)\\
        &= \dot{\phi} + \inner{\tangentS}{\etab_{L}+\etab_{R}, \bb[\Vb_{\mfrak}]} \formPeriod
\end{align*}
by \eqref{eq:scalar_timeder} and \eqref{eq:timederwithu_normal}.
\begin{corollary}
    For all $ \Rb = \rb + \etab_L\otimes\normal + \normal\otimes\etab_R + \phi\normal\otimes\normal \in \tangentR[^2] $, 
   $ \rb\in\tangentS[^2] $, $ \etab_L,\etab_R \in\tangentS $ and $ \phi\in\tangentS[^0] $ holds
    \begin{align} \label{eq:ttensor_Dmat}
        \Dmat\Rb 
             &=   \dot{\rb} - \etab_L\otimes\bb[\Vb_{\mfrak}] - \bb[\Vb_{\mfrak}]\otimes\etab_R 
                + \left( \dot{\phi} + \inner{\tangentS}{\etab_{L}+\etab_{R}, \bb[\Vb_{\mfrak}]} \right)\normal\otimes\normal \notag\\
             &\quad\quad +\left( \dot{\etab}_{L} + \rb\bb[\Vb_{\mfrak}] - \phi\bb[\Vb_{\mfrak}] \right)\otimes\normal
                        +\normal\otimes\left( \dot{\etab}_{R} + \bb[\Vb_{\mfrak}]\rb - \phi\bb[\Vb_{\mfrak}] \right) \formPeriod
    \end{align}
\end{corollary}
To show inner product compatibility of the material derivative, we use $ \dot{\delta}_{AB} =0 $ as for the vector field case.
This gives the following corollary.
\begin{corollary} \label{col:ttensor_innerprodcomp_Dmat}
    The material derivative on 2-tensor fields is compatible with the inner product,
    \ie\ for all $ \Rb_{\alpha} = \rb_{\alpha} + \etab_{\alpha L}\otimes\normal + \normal\otimes\etab_{\alpha R} + \phi_{\alpha}\normal\otimes\normal \in \tangentR[^2] $, 
    with $ \alpha=1,2 $, holds
    \begin{align}
        \overdotinner{\tangentR[^2]}{\Rb_1, \Rb_2} \label{eq:ttensor_inner_Dmat}
            &= \inner{\tangentR[^2]}{ \Dmat\Rb_1 , \Rb_2 } + \inner{\tangentR[^2]}{ \Rb_1 , \Dmat\Rb_2 }\\
            &= \inner{\tangentS[^2]}{ \dot{\rb}_1 , \rb_2} + \inner{\tangentS[^2]}{ \rb_1 , \dot{\rb}_2} 
                + \dot{\phi}_1 \phi_2 + \phi_1 \dot{\phi}_2 \notag\\
            &\quad + \inner{\tangentS}{ \dot{\etab}_{1 L} ,  \etab_{2 L}} + \inner{\tangentS}{ \etab_{1 L} , \dot{\etab}_{2 L}} 
                   + \inner{\tangentS}{ \dot{\etab}_{1 R} ,  \etab_{2 R}} + \inner{\tangentS}{ \etab_{1 R} , \dot{\etab}_{2 R}}  \formPeriod \notag
    \end{align}
\end{corollary}
Since for all $ \Pb\in\tangentR $ is 
$ \Dmat(\Rb\Pb) = \overdot{R^{AB}P_B}\eb_A = (\dot{R}^{AB}P_B + \tensor{\dot{R}}{^A_B}\dot{P}^B)\eb_A$ valid, we obtain the following corollary.
\begin{corollary}\label{col:ttensor_usualprodcomp_Dmat}
    The material derivative is compatible with the 2-tensor-vector product, 
    \ie\ for all $ \Rb = \rb + \etab_L\otimes\normal + \normal\otimes\etab_R + \phi\normal\otimes\normal \in \tangentR[^2] $
    and $ \Pb=\pb+\psi\normal\in\tangentS $, 
   $ \rb\in\tangentS[^2] $, $ \etab_L,\etab_R,\pb \in\tangentS $ and $ \phi,\psi\in\tangentS[^0] $ holds
    \begin{align} \label{eq:ttensor_prodwithvector_Dmat}
        \Dmat(\Rb\Pb) 
             &= (\Dmat\Rb)\Pb + \Rb(\Dmat\Pb) \\
             &= \dot{\rb}\pb + \rb\dot{\pb}
                +\dot{\psi}\etab_{L} + \psi\dot{\etab}_L
                -( \phi\psi + \inner{\tangentS}{\etab_{R}, \pb})\bb[\Vb_{\mfrak}] \notag\\
             &\quad +\left( \dot{\phi}\psi + \phi\dot{\psi} 
                        + \inner{\tangentS}{\dot{\etab}_R , \pb} + \inner{\tangentS}{\etab_R , \dot{\pb}} 
                        +\inner{\tangentS}{\rb\pb + \psi\etab_{L}, \bb[\Vb_{\mfrak}]}\right)\normal \formPeriod \notag
    \end{align}
\end{corollary}

\subsection{Upper-Convected Derivative} \label{sec:Dupp}

In order to obtain the upper-convected derivative, we choose a pullback for the time derivative \eqref{eq:tensor_timeD_general}, 
which adhere to the contravariant material proxy instead of the Cartesian proxy as it is stipulated for the material derivative.
We give the exact definition for the vector and tensor field case in its associated subsections.
In contrast to \cite{NitschkeVoigt_JoGaP_2022}, we use the short naming ``upper-convected''
for ``upper-upper-convected'' or ``fully-upper-convected'', 
since we do not treat any mixed-convected derivative in this paper.

\subsubsection{Vector Fields}

We consider the upper-convected pullback $  \Phi^{*_\sharp}_{t,\tau}:\tangent\R^3\vert_{\surf,t+\tau} \rightarrow \tangent\R^3\vert_{\surf,t}  $
given by 
\begin{align*}
    (\Phi^{*_\sharp}_{t,\tau}\Rb[\para_{\mfrak}]\vert_{t+\tau})(t,y_{\mfrak}^1,y_{\mfrak}^2) 
        &:= r^{i}[\para_{\mfrak}](t+\tau,y_{\mfrak}^1,y_{\mfrak}^2) \partial_i\para_{\mfrak}(t,y_{\mfrak}^1,y_{\mfrak}^2)
            + \phi[\para_{\mfrak}](t+\tau,y_{\mfrak}^1,y_{\mfrak}^2) \normal [\para_{\mfrak}](t,y_{\mfrak}^1,y_{\mfrak}^2)
\end{align*}
for decompositions \eqref{eq:vector_decomposition} of vector fields $ \Rb $.
With this we define the upper-convected derivative by $ \Dupp := \Dt\vert_{\Phi^{*}_{t,\tau}=\Phi^{*_\sharp}_{t,\tau}} $,
\ie\ the time derivative \eqref{eq:tensor_timeD_general} yields
\begin{align} \label{eq:vector_Dupp_material_observer}
    \Dupp\Rb[\para_{\mfrak}]
        &= (\partial_t r^{i}[\para_{\mfrak}]) \partial_i\para_{\mfrak}
            + (\partial_t\phi[\para_{\mfrak}])\normal[\para_{\mfrak}]
\end{align} 
\wrt\ the material observer locally at material events $ (t,y_{\mfrak}^1,y_{\mfrak}^2) $.
Instead of transforming the frame to an arbitrary observer frame, we simply relate this to the material derivative \eqref{eq:vector_Dmat},
where we already know the observer-invariant description.
Term by term, we obtain
\begin{align*}
    (\partial_t r^{i}[\para_{\mfrak}]) \partial_i\para_{\mfrak}
        &= \partial_t ( r^A[\para_{\mfrak}]\eb_A ) - r^{i}[\para_{\mfrak}] \partial_i \Vb_{\mfrak}\\
        &= \Dmat \rb [\para_{\mfrak}] - \Gb[\para_{\mfrak},\Vb_{\mfrak}]\rb[\para_{\mfrak}] 
                - \inner{\tangentS}{\rb[\para_{\mfrak}], \bb[\para_{\mfrak},\Vb_{\mfrak}]} \normal[\para_{\mfrak}]\\
    (\partial_t\phi[\para_{\mfrak}])\normal[\para_{\mfrak}]
        &=  \partial_t( \phi[\para_{\mfrak}]\normal[\para_{\mfrak}] ) - \phi[\para_{\mfrak}]\partial_t \normal[\para_{\mfrak}]
         = \Dmat( \phi\normal)[\para_{\mfrak}] + \phi[\para_{\mfrak}]\bb[\para_{\mfrak}, \Vb_{\mfrak}]
\end{align*}
with \eqref{eq:partialder_V} and \eqref{eq:partialtimeder_normal}.
The first summands are adding up to the material derivative $ \Dmat\Rb[\para_{\mfrak}] $, which has an observer-invariant representation.
The remaining summands are instantaneous and hence are observer-invariant a priori.
Therefore we can express the upper-convected derivative by an arbitrary observer.
This justified to omit the parameterization argument $ \para_{\ofrak} $. 
For instance, we write $ \Dupp\Rb = \Dupp\Rb[\para_{\ofrak}] $ for short. 
Moreover, we could relate $\dot{\rb}$ to the tangential upper-convected derivative by $ \timeLu\rb = \dot{\rb} - \Gb[\Vb_{\mfrak}]\rb\in\tangentS $ 
given in table \ref{tab:vector_tangentialtimederivatives}.
We summarize this in the following corollary under the aid of the tensor field
\begin{align}
    \Gbcal[\Vb_{\mfrak}]  \label{eq:Gbcal_def}
            &:= \Gb[\Vb_{\mfrak}] + \normal \otimes \bb[\Vb_{\mfrak}] - \bb[\Vb_{\mfrak}] \otimes \normal\\ 
            &\phantom{:}= \nabla\vb_{\mfrak} - \vnor\shop + \normal\otimes\left( \nabla\vnor + \shop\vb_{\mfrak} \right) - \left( \nabla\vnor + \shop\vb_{\mfrak} \right)\otimes\normal
                    \in\tangentR[^2] \formPeriod \notag
\end{align} 
\begin{corollary}
    For all $ \Rb = \rb + \phi\normal \in \tangentR $, $ \rb\in\tangentS $ and $ \phi\in\tangentS[^0] $ holds
    \begin{align} \label{eq:vector_Dupp}
        \Dupp\Rb 
            &= \Dmat\Rb - \Gbcal[\Vb_{\mfrak}]\Rb 
             = \timeLu\rb + \dot{\phi}\normal\formPeriod
    \end{align}
\end{corollary}
Note that in contrast to the material or Jaumann derivative, the upper-convected derivative is not compatible with the inner product in general.
Substituting \eqref{eq:vector_Dupp} into \eqref{eq:vector_inner_Dmat} yields
\begin{align*}
    \overdotinner{\tangentR}{\Rb_1, \Rb_2}
          &= \inner{\tangentR}{ \Dupp\Rb_1 , \Rb_2 } + \inner{\tangentR}{ \Rb_1 , \Dupp\Rb_2 }
                + \inner{\tangentR[^{2}]}{ \Gb[\Vb_{\mfrak}] +  \Gb^T[\Vb_{\mfrak}] , \Rb_1 \otimes \Rb_2 }
\end{align*}
for all $ \Rb_1,\Rb_2\in\tangentR $, 
where $ \Gb[\Vb_{\mfrak}] +  \Gb^T[\Vb_{\mfrak}] $ is vanishing if and only if the material carries out a rigid body motion.

For extended vector fields $ \widehat{\Rb}\in\tangent{\surf_h} $, which are sufficing $ \widehat{\Rb}\vert_{\xi=0} = \Rb\in\tangentR $,
we conclude from \eqref{eq:tfl_Dmat} and \eqref{eq:tfl_deformation} that for the upper-convected $ \R^3 $-time derivative
\begin{align*}
   \dot{\widehat{\Rb}} - (\nablahat\widehat{\Vb}_{\mfrak})\widehat{\Rb}
        &\rightarrow \Dupp\Rb
\end{align*}
is valid for $ h \rightarrow 0 $.

\subsubsection{2-Tensor Fields}
We consider the upper-convected pullback $  \Phi^{*_\sharp}_{t,\tau}:\tangent^2\R^3\vert_{\surf,t+\tau} \rightarrow \tangent^2\R^3\vert_{\surf,t}  $
given by 
\begin{align*}
    (\Phi^{*_\sharp}_{t,\tau}\Rb[\para_{\mfrak}]\vert_{t+\tau})(t,y_{\mfrak}^1,y_{\mfrak}^2)
       &:= r^{ij}[\para_{\mfrak}](t+\tau,y_{\mfrak}^1,y_{\mfrak}^2) 
                \partial_i\para_{\mfrak}(t,y_{\mfrak}^1,y_{\mfrak}^2) \otimes \partial_j\para_{\mfrak}(t,y_{\mfrak}^1,y_{\mfrak}^2) \\
       &\quad\quad + \eta_L^i [\para_{\mfrak}](t+\tau,y_{\mfrak}^1,y_{\mfrak}^2) 
                \partial_i\para_{\mfrak}(t,y_{\mfrak}^1,y_{\mfrak}^2) \otimes \normal [\para_{\mfrak}](t,y_{\mfrak}^1,y_{\mfrak}^2)\\
       &\quad\quad + \eta_R^j [\para_{\mfrak}](t+\tau,y_{\mfrak}^1,y_{\mfrak}^2) 
                 \normal [\para_{\mfrak}](t,y_{\mfrak}^1,y_{\mfrak}^2) \otimes \partial_j\para_{\mfrak}(t,y_{\mfrak}^1,y_{\mfrak}^2) \\
       &\quad\quad +\phi[\para_{\mfrak}](t+\tau,y_{\mfrak}^1,y_{\mfrak}^2)
                 \normal [\para_{\mfrak}](t,y_{\mfrak}^1,y_{\mfrak}^2) \otimes \normal [\para_{\mfrak}](t,y_{\mfrak}^1,y_{\mfrak}^2)
\end{align*}
for decompositions \eqref{eq:ttensor_decomposition} of 2-tensor fields $ \Rb $.
With this we define the upper-convected derivative by $ \Dupp := \Dt\vert_{\Phi^{*}_{t,\tau}=\Phi^{*_\sharp}_{t,\tau}} $,
\ie\ the time derivative \eqref{eq:tensor_timeD_general} yields
\begin{align} 
    \Dupp\Rb[\para_{\mfrak}]
        &= (\partial_t r^{ij}[\para_{\mfrak}]) \partial_i\para_{\mfrak} \otimes \partial_j\para_{\mfrak}
            + (\partial_t\phi[\para_{\mfrak}]) \normal [\para_{\mfrak}] \otimes \normal [\para_{\mfrak}] \notag\\
        &\quad\quad + (\partial_t  \eta_L^i [\para_{\mfrak}]) \partial_i\para_{\mfrak} \otimes \normal [\para_{\mfrak}]
                    + (\partial_t  \eta_R^j [\para_{\mfrak}]) \normal [\para_{\mfrak}] \otimes \partial_j\para_{\mfrak}
\label{eq:ttensor_Dupp_material_observer}
\end{align}
\wrt\ the material observer locally at material events $ (t,y_{\mfrak}^1,y_{\mfrak}^2) $.
Similar to the proceeding for vector fields, we just relate this to the material derivative.
Term by term, we obtain
\begin{align*}
    (\partial_t r^{ij}[\para_{\mfrak}]) \partial_i\para_{\mfrak} \otimes \partial_j\para_{\mfrak}
       &= \partial_t\left( r^{AB}[\para_{\mfrak}] \eb_A \otimes\eb_B \right)\\
       &\quad\quad - r^{ij}[\para_{\mfrak}] \left( \tensor{G}{^k_i}[\para_{\mfrak},\Vb_{\mfrak}] \partial_k\para_{\mfrak} 
                                                            + b_i[\para_{\mfrak},\Vb_{\mfrak}]\normal [\para_{\mfrak}] \right) 
                \otimes \partial_j\para_{\mfrak} \\
       &\quad\quad -r^{ij}[\para_{\mfrak}] \partial_i\para_{\mfrak} 
                \otimes\left( \tensor{G}{^k_j}[\para_{\mfrak},\Vb_{\mfrak}] \partial_k\para_{\mfrak} + b_j[\para_{\mfrak},\Vb_{\mfrak}]\normal [\para_{\mfrak}] \right)\\
       &= \Dmat\rb[\para_{\mfrak}] 
                - \Gb[\para_{\mfrak},\Vb_{\mfrak}]\rb[\para_{\mfrak}] - \rb[\para_{\mfrak}]\Gb^{T}[\para_{\mfrak},\Vb_{\mfrak}] \\
       &\quad\quad - \normal [\para_{\mfrak}] \otimes \bb[\para_{\mfrak},\Vb_{\mfrak}] \rb[\para_{\mfrak}]
                    - \rb[\para_{\mfrak}]\bb[\para_{\mfrak},\Vb_{\mfrak}]\otimes\normal [\para_{\mfrak}]\\
   (\partial_t  \eta_L^i [\para_{\mfrak}]) \partial_i\para_{\mfrak} \otimes \normal [\para_{\mfrak}]
        &= \Dupp\etab_{L}[\para_{\mfrak}] \otimes \normal [\para_{\mfrak}]
         = \left( \Dmat\etab_{L}[\para_{\mfrak}] -\Gbcal[\para_{\mfrak},\Vb_{\mfrak}]\etab_{L}[\para_{\mfrak}] \right)\otimes \normal [\para_{\mfrak}]\\
        &= \Dmat(\etab_{L} \otimes \normal)[\para_{\mfrak}] + \etab_{L}[\para_{\mfrak}]\otimes\bb[\para_{\mfrak},\Vb_{\mfrak}]
            - \Gb[\para_{\mfrak},\Vb_{\mfrak}]\etab_{L}[\para_{\mfrak}]\otimes \normal [\para_{\mfrak}]\\
        &\quad\quad - \inner{\tangentS}{ \etab_{L}[\para_{\mfrak}], \bb[\para_{\mfrak},\Vb_{\mfrak}] } \normal [\para_{\mfrak}] \otimes \normal [\para_{\mfrak}]\\
    (\partial_t  \eta_R^j [\para_{\mfrak}]) \normal [\para_{\mfrak}] \otimes \partial_j\para_{\mfrak}
         &= \left( (\partial_t  \eta_R^i [\para_{\mfrak}]) \partial_i\para_{\mfrak} \otimes \normal [\para_{\mfrak}] \right)^{T}\\
         &= \Dmat(\normal\otimes\etab_{R})[\para_{\mfrak}]
                 + \bb[\para_{\mfrak},\Vb_{\mfrak}] \otimes \etab_{R}[\para_{\mfrak}]  -  \normal [\para_{\mfrak}] \otimes \Gb[\para_{\mfrak},\Vb_{\mfrak}]\etab_R[\para_{\mfrak}] \\
         &\quad\quad - \inner{\tangentS}{\etab_R[\para_{\mfrak}],\bb[\para_{\mfrak},\Vb_{\mfrak}]}\normal [\para_{\mfrak}] \otimes \normal [\para_{\mfrak}]\\
    (\partial_t\phi[\para_{\mfrak}]) \normal [\para_{\mfrak}] \otimes \normal [\para_{\mfrak}]
        &= \dot{\phi}[\para_{\mfrak}] \normal [\para_{\mfrak}] \otimes \normal [\para_{\mfrak}] \\
        &= \Dmat( \phi\normal\otimes\normal )[\para_{\mfrak}] 
                + \phi[\para_{\mfrak}]\left( \bb[\para_{\mfrak},\Vb_{\mfrak}] \otimes \normal [\para_{\mfrak}] 
                                                + \normal [\para_{\mfrak}] \otimes \bb[\para_{\mfrak},\Vb_{\mfrak}] \right)
\end{align*}
due to \eqref{eq:partialder_V}, \eqref{eq:vector_Dupp_material_observer}, \eqref{eq:vector_Dupp} and \eqref{eq:ttensor_Dmat}.
The first summands are adding up to the material derivative $ \Dmat\Rb[\para_{\mfrak}] $, which has an observer-invariant representation.
The remaining summands are instantaneous and hence observer-invariant a priori.
Therefore we can express the upper-convected derivative by an arbitrary observer.
This justifies to omit observer arguments in square brackets.
\begin{corollary}
    For all $ \Rb = \rb + \etab_L\otimes\normal + \normal\otimes\etab_R + \phi\normal\otimes\normal \in \tangentR[^2] $, 
   $ \rb\in\tangentS[^2] $, $ \etab_L,\etab_R \in\tangentS $ and $ \phi\in\tangentS[^0] $ holds
    \begin{align} \label{eq:ttensor_Dupp}
        \Dupp\Rb 
             &=  \Dmat\Rb - \Gbcal[\Vb_{\mfrak}]\Rb - \Rb \Gbcal^T[\Vb_{\mfrak}]
             = \timeLuu\rb + \timeLu\etab_{L}\otimes\normal + \normal\otimes\timeLu\etab_R + \dot{\phi}\normal\otimes\normal \formPeriod
    \end{align}
\end{corollary}
Tangential upper-convected derivatives $ \timeLuu $ and $ \timeLu $ on tangential tensor fields are given 
in table \ref{tab:ttensor_tangentialtimederivatives} and \ref{tab:vector_tangentialtimederivatives}.
Note that in contrast to the material or Jaumann derivative, the upper-convected derivative is neither compatible with the inner nor the tensor product in general.
Substituting \eqref{eq:ttensor_Dupp} into \eqref{eq:ttensor_inner_Dmat}, 
\resp\ \eqref{eq:ttensor_Dupp} and \eqref{eq:vector_Dupp} into \eqref{eq:ttensor_prodwithvector_Dmat},  yields
\begin{align*}
    \overdotinner{\tangentR[^2]}{\Rb_1, \Rb_2}
          &= \inner{\tangentR[^2]}{ \Dupp\Rb_1 , \Rb_2 } + \inner{\tangentR[^2]}{ \Rb_1 , \Dupp\Rb_2 }
                + \inner{\tangentR[^{2}]}{ \Gb[\Vb_{\mfrak}] +  \Gb^T[\Vb_{\mfrak}] , \Rb_1 \Rb_2^T + \Rb_1^T \Rb_2 }\\
    \Dupp(\Rb\Pb)
        &= (\Dupp\Rb)\Pb + \Rb(\Dupp\Pb) + \Rb\left( \Gb[\Vb_{\mfrak}] +  \Gb^T[\Vb_{\mfrak}] \right)\Pb
\end{align*}
for all $ \Rb_1,\Rb_2,\Rb\in\tangentR[^2] $ and $ \Pb\in\tangentR $, 
where $ \Gb[\Vb_{\mfrak}] +  \Gb^T[\Vb_{\mfrak}] $ is vanishing if and only if the material carries out a rigid body motion.

For extended 2-tensor fields $ \widehat{\Rb}\in\tangent[^2]{\surf_h} $, which are sufficing $ \widehat{\Rb}\vert_{\xi=0} = \Rb\in\tangentR[^2] $,
we conclude from \eqref{eq:tfl_Dmat} and \eqref{eq:tfl_deformation} that that for the upper-convected $ \R^3 $-time derivative
\begin{align*}
    \dot{\widehat{\Rb}} - (\nablahat\widehat{\Vb}_{\mfrak})\widehat{\Rb} - \widehat{\Rb}(\nablahat\widehat{\Vb}_{\mfrak})^T
        &\rightarrow \Dupp\Rb
\end{align*}
is valid for $ h \rightarrow 0 $.

\subsection{Lower-Convected Derivative} \label{sec:Dlow}

In order to obtain the lower-convected derivative, we choose a pullback for the time derivative \eqref{eq:tensor_timeD_general}, 
which adhere to the covariant material proxy instead of the contravariant material proxy as it is stipulated for the upper-convected derivative.
We give the exact definition for the vector and tensor field case in its associated subsections.
In contrast to \cite{NitschkeVoigt_JoGaP_2022}, we use the short naming ``lower-convected''
for ``lower-lower-convected'' or ``fully-lower-convected'', 
since we do not treat any mixed-convected derivative in this paper.

\subsubsection{Vector Fields}

We consider the lower-convected pullback 
$ \Phi^{*_\flat}_{t,\tau}:\tangent\R^3\vert_{\surf,t+\tau} \rightarrow \tangent\R^3\vert_{\surf,t} $ 
given by 
\begin{align*}
    (\Phi^{*_\flat}_{t,\tau}\Rb[\para_{\mfrak}]\vert_{t+\tau})(t,y_{\mfrak}^1,y_{\mfrak}^2) 
        :=  r_{i}[\para_{\mfrak}](t+\tau,y_{\mfrak}^1,y_{\mfrak}^2) \partial^i\para_{\mfrak}(t,y_{\mfrak}^1,y_{\mfrak}^2)
            + \phi[\para_{\mfrak}](t+\tau,y_{\mfrak}^1,y_{\mfrak}^2) \normal [\para_{\mfrak}](t,y_{\mfrak}^1,y_{\mfrak}^2)
\end{align*}
for decompositions \eqref{eq:vector_decomposition} of vector fields $ \Rb $,
where $ \partial^i\para_{\mfrak} := g_{\mfrak}^{ij}\partial_j\para_{\mfrak} $ at all $ (t,y_{\mfrak}^1,y_{\mfrak}^2) $.
With this we define the upper-convected derivative by $ \Dlow := \Dt\vert_{\Phi^{*}_{t,\tau}=\Phi^{*_\flat}_{t,\tau}} $,
\ie\ the time derivative \eqref{eq:tensor_timeD_general} yields
\begin{align*} 
    \Dlow\Rb[\para_{\mfrak}]
        &= g_{\mfrak}^{ij} (\partial_t r_{j}[\para_{\mfrak}]) \partial_i\para_{\mfrak}
            + (\partial_t\phi[\para_{\mfrak}])\normal[\para_{\mfrak}]
\end{align*} 
\wrt\ the material observer locally at material events $ (t,y_{\mfrak}^1,y_{\mfrak}^2) $.
With \eqref{eq:vector_partialtimeder_covar}, this is relatable to the upper-convected derivative \eqref{eq:vector_Dupp_material_observer} by
\begin{align*}
    \Dlow\Rb[\para_{\mfrak}]
        &= \Dupp\Rb[\para_{\mfrak}] + (\Gb[\para_{\mfrak},\Vb_{\mfrak}]+\Gb^T[\para_{\mfrak},\Vb_{\mfrak}])\rb[\para_{\mfrak}] \formPeriod
\end{align*}
This expression can be represented by an observer-invariant formulation.
Therefore, we omit the observer argument in square brackets. 
Moreover, since all normal parts of $ \Gbcal[\Vb_{\mfrak}] $ \eqref{eq:Gbcal_def} are antisymmetric, 
\ie\ it is $  \Gbcal[\Vb_{\mfrak}] +  \Gbcal^T[\Vb_{\mfrak}] = \Gb[\Vb_{\mfrak}] +  \Gb^T[\Vb_{\mfrak}] $ valid,
we conclude the following corollary.
\begin{corollary}
    For all $ \Rb = \rb + \phi\normal \in \tangentR $, $ \rb\in\tangentS $ and $ \phi\in\tangentS[^0] $ holds
    \begin{align} \label{eq:vector_Dlow}
        \Dlow\Rb 
            &= \Dmat\Rb + \Gbcal^T[\Vb_{\mfrak}]\Rb 
             = \timeLl\rb + \dot{\phi}\normal\formPeriod
    \end{align}
\end{corollary}
The tangential lower-convected derivative $ \timeLl $ on tangential vector fields is given 
in table \ref{tab:vector_tangentialtimederivatives}.
Note that in contrast to the material or Jaumann derivative, the lower-convected derivative is not compatible with the inner product in general.
Substituting \eqref{eq:vector_Dlow} into \eqref{eq:vector_inner_Dmat} yields
\begin{align*}
    \overdotinner{\tangentR}{\Rb_1, \Rb_2}
          &= \inner{\tangentR}{ \Dlow\Rb_1 , \Rb_2 } + \inner{\tangentR}{ \Rb_1 , \Dlow\Rb_2 }
                - \inner{\tangentR[^{2}]}{ \Gb[\Vb_{\mfrak}] +  \Gb^T[\Vb_{\mfrak}] , \Rb_1 \otimes \Rb_2 }
\end{align*}
for all $ \Rb_1,\Rb_2\in\tangentR $, 
where $ \Gb[\Vb_{\mfrak}] +  \Gb^T[\Vb_{\mfrak}] $ is vanishing if and only if the material carries out a rigid body motion.

For extended vector fields $ \widehat{\Rb}\in\tangent{\surf_h} $, which are sufficing $ \widehat{\Rb}\vert_{\xi=0} = \Rb\in\tangentR $,
we conclude from \eqref{eq:tfl_Dmat} and \eqref{eq:tfl_deformation} that for the lower-convected $ \R^3 $-time derivative
\begin{align*}
     \dot{\widehat{\Rb}} + (\nablahat\widehat{\Vb}_{\mfrak})^T\widehat{\Rb}
        &\rightarrow \Dlow\Rb
\end{align*}
is valid for $ h \rightarrow 0 $.

\subsubsection{2-Tensor Fields}
We consider the lower-convected pullback 
$ \Phi^{*_\flat}_{t,\tau}:\tangent^2\R^3\vert_{\surf,t+\tau} \rightarrow \tangent^2\R^3\vert_{\surf,t} $ 
given by
\begin{align*}
    (\Phi^{*_\flat}_{t,\tau}\Rb[\para_{\mfrak}]\vert_{t+\tau})(t,y_{\mfrak}^1,y_{\mfrak}^2)
      &=r_{ij}[\para_{\mfrak}](t+\tau,y_{\mfrak}^1,y_{\mfrak}^2) 
               \partial^i\para_{\mfrak}(t,y_{\mfrak}^1,y_{\mfrak}^2) \otimes \partial^j\para_{\mfrak}(t,y_{\mfrak}^1,y_{\mfrak}^2) \\
      &\quad\quad + \eta_{L i} [\para_{\mfrak}](t+\tau,y_{\mfrak}^1,y_{\mfrak}^2) 
               \partial^i\para_{\mfrak}(t,y_{\mfrak}^1,y_{\mfrak}^2) \otimes \normal [\para_{\mfrak}](t,y_{\mfrak}^1,y_{\mfrak}^2)\\
      &\quad\quad + \eta_{R j} [\para_{\mfrak}](t+\tau,y_{\mfrak}^1,y_{\mfrak}^2) 
                \normal [\para_{\mfrak}](t,y_{\mfrak}^1,y_{\mfrak}^2) \otimes \partial^j\para_{\mfrak}(t,y_{\mfrak}^1,y_{\mfrak}^2) \\
      &\quad\quad +\phi[\para_{\mfrak}](t+\tau,y_{\mfrak}^1,y_{\mfrak}^2)
                \normal [\para_{\mfrak}](t,y_{\mfrak}^1,y_{\mfrak}^2) \otimes \normal [\para_{\mfrak}](t,y_{\mfrak}^1,y_{\mfrak}^2)
\end{align*}
for decompositions \eqref{eq:ttensor_decomposition} of 2-tensor fields $ \Rb $.
With this we define the upper-convected derivative by $ \Dlow := \Dt\vert_{\Phi^{*}_{t,\tau}=\Phi^{*_\flat}_{t,\tau}} $,
\ie\ the time derivative \eqref{eq:tensor_timeD_general} yields
\begin{align*}
    \Dlow\Rb[\para_{\mfrak}]
        &= (\partial_t r_{ij}[\para_{\mfrak}]) g_{\mfrak}^{ik} g_{\mfrak}^{jl} \partial_k\para_{\mfrak} \otimes \partial_l\para_{\mfrak}
            + (\partial_t\phi[\para_{\mfrak}]) \normal [\para_{\mfrak}] \otimes \normal [\para_{\mfrak}] \\
        &\quad\quad + (\partial_t  \eta_{L i} [\para_{\mfrak}]) g_{\mfrak}^{ik} \partial_k\para_{\mfrak} \otimes \normal [\para_{\mfrak}]
                    + (\partial_t  \eta_{R j} [\para_{\mfrak}])g_{\mfrak}^{jl} \normal [\para_{\mfrak}] \otimes \partial_l\para_{\mfrak}
\end{align*}
\wrt\ the material observer locally at material events $ (t,y_{\mfrak}^1,y_{\mfrak}^2) $.
With \eqref{eq:ttensor_partialtimeder_covar} and \eqref{eq:vector_partialtimeder_covar}, this is relatable to the upper-convected derivative \eqref{eq:ttensor_Dupp_material_observer} by
\begin{align*}
    \Dlow\Rb[\para_{\mfrak}]
        &= \Dupp\Rb[\para_{\mfrak}] 
                + 2\Sb[\para_{\mfrak},\Vb_{\mfrak}]\rb[\para_{\mfrak}] + 2\rb[\para_{\mfrak}] \Sb[\para_{\mfrak},\Vb_{\mfrak}]\\
        &\quad\quad + 2\Sb[\para_{\mfrak},\Vb_{\mfrak}] \etab_{L}[\para_{\mfrak}] \otimes \normal[\para_{\mfrak}]
                    + 2\normal[\para_{\mfrak}] \otimes \Sb[\para_{\mfrak},\Vb_{\mfrak}] \etab_{R}[\para_{\mfrak}] \formComma
\end{align*}
where $ 2\Sb[\para_{\mfrak},\Vb_{\mfrak}] := \Gb[\para_{\mfrak},\Vb_{\mfrak}] + \Gb^T[\para_{\mfrak},\Vb_{\mfrak}] = \Gbcal[\Vb_{\mfrak}] +  \Gbcal^T[\Vb_{\mfrak}]$.
This expression can be represented by an observer-invariant formulation.
Therefore, we omit the observer argument in square brackets and conclude the following corollary.
\begin{corollary}
    For all $ \Rb = \rb + \etab_L\otimes\normal + \normal\otimes\etab_R + \phi\normal\otimes\normal \in \tangentR[^2] $, 
   $ \rb\in\tangentS[^2] $, $ \etab_L,\etab_R \in\tangentS $ and $ \phi\in\tangentS[^0] $ holds
    \begin{align} \label{eq:ttensor_Dlow}
        \Dlow\Rb 
             &=  \Dmat\Rb + \Gbcal^T[\Vb_{\mfrak}]\Rb + \Rb \Gbcal[\Vb_{\mfrak}]
             = \timeLll\rb + \timeLl\etab_{L}\otimes\normal + \normal\otimes\timeLl\etab_R + \dot{\phi}\normal\otimes\normal \formPeriod
    \end{align}
\end{corollary}
Tangential lower-convected derivatives $ \timeLll $ and $ \timeLl $ on tangential tensor fields are given 
in table \ref{tab:ttensor_tangentialtimederivatives} and \ref{tab:vector_tangentialtimederivatives}.
Note that in contrast to the material or Jaumann derivative, the upper-convected derivative is neither compatible with the inner nor the tensor product in general.
Substituting \eqref{eq:ttensor_Dlow} into \eqref{eq:ttensor_inner_Dmat}, 
\resp\ \eqref{eq:ttensor_Dlow} and \eqref{eq:vector_Dlow} into \eqref{eq:ttensor_prodwithvector_Dmat},  yields
\begin{align*}
    \overdotinner{\tangentR[^2]}{\Rb_1, \Rb_2}
          &= \inner{\tangentR[^2]}{ \Dlow\Rb_1 , \Rb_2 } + \inner{\tangentR[^2]}{ \Rb_1 , \Dlow\Rb_2 }
                - \inner{\tangentR[^{2}]}{ \Gb[\Vb_{\mfrak}] +  \Gb^T[\Vb_{\mfrak}] , \Rb_1 \Rb_2^T + \Rb_1^T \Rb_2 }\\
 \Dlow(\Rb\Pb)
        &= (\Dlow\Rb)\Pb + \Rb(\Dlow\Pb) - \Rb\left( \Gb[\Vb_{\mfrak}] +  \Gb^T[\Vb_{\mfrak}] \right)\Pb
\end{align*}
for all $ \Rb_1,\Rb_2,\Rb\in\tangentR[^2] $ and $ \Pb\in\tangentR $, 
where $ \Gb[\Vb_{\mfrak}] +  \Gb^T[\Vb_{\mfrak}] $ is vanishing if and only if the material carries out a rigid body motion.

For extended 2-tensor fields $ \widehat{\Rb}\in\tangent[^2]{\surf_h} $, which are sufficing $ \widehat{\Rb}\vert_{\xi=0} = \Rb\in\tangentR[^2] $,
we conclude from \eqref{eq:tfl_Dmat} and \eqref{eq:tfl_deformation} that for the lower-convected $ \R^3 $-time derivative
\begin{align*}
    \dot{\widehat{\Rb}} + (\nablahat\widehat{\Vb}_{\mfrak})^T\widehat{\Rb} + \widehat{\Rb}(\nablahat\widehat{\Vb}_{\mfrak})
        &\rightarrow \Dlow\Rb
\end{align*}
is valid for $ h \rightarrow 0 $.

\subsection{Jaumann Derivative} \label{sec:Djau}

For the sake of simplicity, we define the Jaumann derivative by
\begin{align}\label{eq:tensor_Djau}
\Djau\Rb := \frac{1}{2}\left( \Dupp\Rb + \Dlow\Rb \right)
\end{align}
for all $ \Rb\in\tangentR[^n] $ instead stipulating a pullback  $ \Phi^{*_\jau}_{t,\tau}$, see the discussion in section \ref{sec:discussion_Jau_approach} for more details.
Therefore, the Jaumann derivative is observer-invariant a priori, since the upper- and lower-convected are observer-invariant.
Note that the Jaumann derivative is also often named corotational derivative.

\subsubsection{Vector Fields}

By defining the skew-symmetric tensor field
\begin{align}
    \Abcal[\Vb_{\mfrak}] \label{eq:Abcal_def}
        &:= \frac{\Gbcal[\Vb_{\mfrak}] - \Gbcal^T[\Vb_{\mfrak}]}{2}
          = \frac{\nabla\vb_{\mfrak} - (\nabla\vb_{\mfrak})^T}{2} + \normal\otimes\left( \nabla\vnor + \shop\vb_{\mfrak} \right) - \left( \nabla\vnor + \shop\vb_{\mfrak} \right)\otimes\normal
          \in\tangentR[^2] \formComma
\end{align}
we deduce the following corollary by \eqref{eq:tensor_Djau}, \eqref{eq:vector_Dupp} and \eqref{eq:vector_Dlow}.
\begin{corollary}
    For all $ \Rb = \rb + \phi\normal \in \tangentR $, $ \rb\in\tangentS $ and $ \phi\in\tangentS[^0] $ holds
    \begin{align} \label{eq:vector_Djau}
        \Djau\Rb 
            &= \Dmat\Rb - \Abcal[\Vb_{\mfrak}]\Rb 
             = \timeJ\rb + \dot{\phi}\normal\formPeriod
    \end{align}
\end{corollary}
The tangential Jaumann derivative $ \timeJ $ on tangential vector fields is given 
in table \ref{tab:vector_tangentialtimederivatives}.
Since $ \Abcal[\Vb_{\mfrak}] $ is skew-symmetric, substituting \eqref{eq:vector_Djau} into \eqref{eq:vector_inner_Dmat} yields the following corollary.
\begin{corollary} \label{col:vector_innerprodcomp_Djau}
    The Jaumann derivative on vector fields is compatible with the inner product,
    \ie\ for all $ \Rb_1 = \rb_1 + \phi_1\normal,\Rb_2 = \rb_2 + \phi_2\normal\in\tangentR $ holds
    \begin{align*}
        \overdotinner{\tangentR}{\Rb_1, \Rb_2}
            &= \inner{\tangentR}{ \Djau\Rb_1 , \Rb_2 } + \inner{\tangentR}{ \Rb_1 , \Djau\Rb_2 }\\
            &= \inner{\tangentS}{ \timeJ\rb_1 , \rb_2} + \inner{\tangentS}{ \rb_1 , \timeJ\rb_2} 
                + \dot{\phi}_1 \phi_2 + \phi_1 \dot{\phi}_2 \formPeriod \notag
    \end{align*}
\end{corollary}

For extended vector fields $ \widehat{\Rb}\in\tangent{\surf_h} $, which are sufficing $ \widehat{\Rb}\vert_{\xi=0} = \Rb\in\tangentR $,
we conclude from \eqref{eq:tfl_Dmat} and \eqref{eq:tfl_deformation} that for the Jaumann $ \R^3 $-time derivative
\begin{align*}
    \dot{\widehat{\Rb}} - \left(\nablahat\widehat{\Vb}_{\mfrak}-(\nablahat\widehat{\Vb}_{\mfrak})^T \right)\widehat{\Rb}
        &\rightarrow \Djau\Rb
\end{align*}
is valid for $ h \rightarrow 0 $.

\subsubsection{2-Tensor Fields}

Using the tensor field $ \Abcal[\Vb_{\mfrak}] $ \eqref{eq:Abcal_def},
we deduce the following corollary by \eqref{eq:tensor_Djau}, \eqref{eq:ttensor_Dupp} and \eqref{eq:ttensor_Dlow}.
\begin{corollary}
    For all $ \Rb = \rb + \etab_L\otimes\normal + \normal\otimes\etab_R + \phi\normal\otimes\normal \in \tangentR[^2] $, 
   $ \rb\in\tangentS[^2] $, $ \etab_L,\etab_R \in\tangentS $ and $ \phi\in\tangentS[^0] $ holds
    \begin{align} \label{eq:ttensor_Djau}
        \Djau\Rb 
             &=  \Dmat\Rb - \Abcal[\Vb_{\mfrak}]\Rb + \Rb \Abcal[\Vb_{\mfrak}]
             = \timeJ\rb + \timeJ\etab_{L}\otimes\normal + \normal\otimes\timeJ\etab_R + \dot{\phi}\normal\otimes\normal \formPeriod
    \end{align}
\end{corollary}
The tangential Jaumann derivative $ \timeJ $ is given 
in table \ref{tab:ttensor_tangentialtimederivatives} for tangential 2-tensor fields and in table \ref{tab:vector_tangentialtimederivatives} for tangential vector fields.
Since $ \Abcal[\Vb_{\mfrak}] $ is skew-symmetric, substituting \eqref{eq:ttensor_Djau} into \eqref{eq:ttensor_inner_Dmat} yields the following corollary.
\begin{corollary}\label{col:ttensor_innerprodcomp_Djau}
    The Jaumann derivative on 2-tensor fields is compatible with the inner product,
    \ie\ for all $ \Rb_{\alpha} = \rb_{\alpha} + \etab_{\alpha L}\otimes\normal + \normal\otimes\etab_{\alpha R} + \phi_{\alpha}\normal\otimes\normal \in \tangentR[^2] $, 
    with $ \alpha=1,2 $, holds
    \begin{align*}
        \overdotinner{\tangentR[^2]}{\Rb_1, \Rb_2}
            &= \inner{\tangentR[^2]}{ \Djau\Rb_1 , \Rb_2 } + \inner{\tangentR[^2]}{ \Rb_1 , \Djau\Rb_2 }\\
            &= \inner{\tangentS[^2]}{ \timeJ\rb_1 , \rb_2} + \inner{\tangentS[^2]}{ \rb_1 , \timeJ\rb_2} 
                + \dot{\phi}_1 \phi_2 + \phi_1 \dot{\phi}_2 \notag\\
            &\quad + \inner{\tangentS}{ \timeJ\etab_{1 L} ,  \etab_{2 L}} + \inner{\tangentS}{ \etab_{1 L} , \timeJ\etab_{2 L}} 
                   + \inner{\tangentS}{ \timeJ\etab_{1 R} ,  \etab_{2 R}} + \inner{\tangentS}{ \etab_{1 R} , \timeJ\etab_{2 R}}  \formPeriod \notag
    \end{align*}
\end{corollary}
Substituting \eqref{eq:ttensor_Djau} and \eqref{eq:vector_Djau} into \eqref{eq:ttensor_prodwithvector_Dmat} results in the following corollary.
\begin{corollary}\label{col:ttensor_usualprodcomp_Djau}
    The Jaumann derivative is compatible with the 2-tensor-vector product, 
    \ie\ for all $ \Rb = \rb + \etab_L\otimes\normal + \normal\otimes\etab_R + \phi\normal\otimes\normal \in \tangentR[^2] $
    and $ \Pb=\pb+\psi\normal\in\tangentS $, 
   $ \rb\in\tangentS[^2] $, $ \etab_L,\etab_R,\pb \in\tangentS $ and $ \phi,\psi\in\tangentS[^0] $ holds
    \begin{align} \label{eq:ttensor_prodwithvector_Djau}
        \Djau(\Rb\Pb) 
             &= (\Djau\Rb)\Pb + \Rb(\Djau\Pb) \\
             &= (\timeJ\rb)\pb + \rb(\timeJ\pb) 
                + \psi\timeJ\etab_{L} + \dot{\psi}\etab_{L}
                +\left( \dot{\phi}\psi + \phi\dot{\psi} 
                        +\inner{\tangentS}{ \timeJ\etab_{R}, \pb } + \inner{\tangentS}{ \etab_{R}, \timeJ\pb }\right)\normal\formPeriod \notag
    \end{align}
\end{corollary}

For extended 2-tensor fields $ \widehat{\Rb}\in\tangent[^2]{\surf_h} $, which are sufficing $ \widehat{\Rb}\vert_{\xi=0} = \Rb\in\tangentR[^2] $,
we conclude from \eqref{eq:tfl_Dmat} and \eqref{eq:tfl_deformation} that for the Jaumann $ \R^3 $-time derivative
\begin{align*}
    \dot{\widehat{\Rb}} - \left(\nablahat\widehat{\Vb}_{\mfrak}-(\nablahat\widehat{\Vb}_{\mfrak})^T \right)\widehat{\Rb}
                                      + \widehat{\Rb} \left(\nablahat\widehat{\Vb}_{\mfrak}-(\nablahat\widehat{\Vb}_{\mfrak})^T \right)
        &\rightarrow \Djau\Rb
\end{align*}
is valid for $ h \rightarrow 0 $.

\subsubsection{Discussion of Approach \eqref{eq:tensor_Djau}}\label{sec:discussion_Jau_approach}

    The choice of a pullback in \eqref{eq:tensor_timeD_general} is indeed sufficient to determine the associated time derivative,
    but it is not necessary, \ie\ we can find other pullbacks, which define the same time derivative. 
    The fully Taylor expansion
    \begin{align*}
        (\Phi^*_{t,\tau}\Rb[\para_{\mfrak}]\vert_{t+\tau})(t,y_{\mfrak}^1,y_{\mfrak}^2)
            = \sum_{\alpha=0}^{\infty}  \frac{\tau^{\alpha}}{\alpha !}(\Dt[(\alpha)]\Rb)[\para_{\mfrak}](t,y_{\mfrak}^1,y_{\mfrak}^2)
    \end{align*}
    at $ \tau=0 $ might be the easiest way to see this, where $ \Dt[(\alpha)] $ is a time derivative of $ \alpha $th order.
    The reason we mention this is that we stipulate the identity \eqref{eq:tensor_Djau} to define the Jaumann derivative in relation to the upper- and lower-convected derivative.
    Indeed, averaging the associated pullbacks to $ \Phi^{*_\jau}_{t,\tau} := \frac{1}{2}\left( \Phi^{*_\sharp}_{t,\tau} + \Phi^{*_\flat}_{t,\tau} \right) $ in the same way
    would be sufficient to obtain the Jaumann derivative \eqref{eq:tensor_Djau} also by $   \Djau = \Dt\vert_{\Phi^{*}_{t,\tau}=\Phi^{*_\jau}_{t,\tau}}  $. 
    However, despite this proceeding would be reasoned, it is not very intuitive. 
    More tangible would be a pullback structurally given by 
    $ \Phi^{*_\jau}_{t,\tau}\Rb[\para_{\mfrak}]\vert_{t+\tau} := \Omegab_{t,\tau}^{T}[\para_{\mfrak}]\Rb[\para_{\mfrak}]\vert_{t+\tau}  $ for vector fields,
    \resp\ $ \Phi^{*_\jau}_{t,\tau}\Rb[\para_{\mfrak}]\vert_{t+\tau} 
        := \Omegab_{t,\tau}^{T}[\para_{\mfrak}]\Rb[\para_{\mfrak}]\vert_{t+\tau} + \Rb[\para_{\mfrak}]\vert_{t+\tau}\Omegab_{t,\tau}[\para_{\mfrak}]   $
    for 2-tensor fields,
    where $ \Omegab_{t,\tau}[\para_{\mfrak}]\in\tangentR[^2] $ is the rotation tensor, which rotate every local tangential plane and normal according to the material deformation
    $ \surf\vert_{t} \rightarrow \surf\vert_{t+\tau} $.
    It is only ensured that this pullback equals the former pullback first orderly \wrt\ the Taylor expansion above.
    For the sake of simplicity we decided to approach the Jaumann derivative by \eqref{eq:tensor_Djau} rather than determining  $ \Omegab_{t,\tau}^{-1} $ and its necessary derivatives.

\subsection{Q-Tensor Fields} \label{sec:qtensor}

\subsubsection{General Q-Tensor Fields}

Beside the orthogonal decomposition 
$ \tangentR[^2] = \tangentS[^2] \oplus (\tangentS\otimes\normal) \oplus (\normal\otimes\tangentS) \oplus (\tangentS[^0]\normal\otimes\normal) $ 
realized by \eqref{eq:ttensor_decomposition}, there is another useful orthogonal decomposition for 2-tensor fields, namely
$ \tangentR[^2] = (\tangentS[^0]\Id)\oplus\tangentAR\oplus\tangentQR $,
where $ \Id $ is the Euclidean identity tensor fields, \eg\ implemented by $ \Id = \delta^{AB}\eb_A\otimes\eb_B $ \wrt\ a Cartesian frame,
$ \tangentAR := \{ \Ab\in\tangentR[^2] : \Ab = -\Ab^T \} $ is the space of skew-symmetric tensor fields and 
\begin{align*}
    \tangentQR &:= \left\{ \Qb\in\tangentR[^2] : \Qb = \Qb^T \text{ and } \Tr\Qb=0 \right\}
\end{align*}
is the space of symmetric and trace-free tensor fields, also called Q-tensor fields.
In this section we examine the latter in context of time derivatives in more detail.

To describe $ \tangentQR $, which established a 5-dimensional vector bundle on $ \surf $, by tangential quantities,
we introduce the orthogonal decomposition
\begin{align}\label{eq:qtensor_decomposition}
    \Qb &= \Qbcal\left[\qb,\etab,\beta\right]
        := \qb + \etab\otimes\normal + \normal\otimes\etab + \beta\left( \normal\otimes\normal - \frac{1}{2}\IdS \right) \in \tangentQR\formComma
\end{align} 
where $ \qb\in\tangentQS := \{ \qb\in\tangentS[^2] : \qb=\qb^T \text{ and } \Tr\qb=0 \}$ is a tangential Q-tensor field,
and $ \etab\in\tangentS $ and $ \beta=\tangentS[^0] $ are determined uniquely for all $ \Qb\in\tangentQR $.
$ \IdS $ is the tangential identity tensor fields, \eg\ implemented by $ \IdS = g^{ij}\partial_i\para\otimes\partial_j\para $ \wrt\ the local tangential frame
or $ \IdS = (\delta^{AB} - \normalC^A\otimes\normalC^B )\eb_A\otimes\eb_B$ \wrt\ a Cartesian frame.
This decomposition is consistent to decomposition \eqref{eq:ttensor_decomposition} for $ \Rb = \Qb $, $ \rb=\qb-\frac{\beta}{2}\IdS $, 
$ \etab_{L} = \etab_{R} = \etab $ and $ \phi=\beta $.
Therefore, \eqref{eq:ttensor_Dmat}, \eqref{eq:ttensor_Djau}, $\dot{\qb}\in\tangentQS$, $ \timeJ\qb\in\tangentQS $ and $ \dot{\IdS} = \timeJ\IdS = 0 $ yields
\begin{align}\label{eq:qtensor_Dmat}
    \Dmat\Qb
        &= \Qbcal\left[\dot{\qb} - 2\projQS(\etab\otimes\bb[\Vb_{\mfrak}]), 
                    \dot{\etab} + \qb\bb[\Vb_{\mfrak}] - \frac{3\beta}{2}\bb[\Vb_{\mfrak}],
                    \dot{\beta} + 2\inner{\tangentS}{\etab,\bb[\Vb_{\mfrak}]}\right] \formComma\\
    \Djau\Qb\label{eq:qtensor_Djau}
        &= \Qbcal\left[\timeJ\qb,\timeJ\etab,\dot{\beta}\right] \formComma
\end{align}
where $ \projQS: \tangentS[^2]\rightarrow\tangentQS $ is the orthogonal projection given by $ \projQS\rb=\frac{1}{2}(\rb+\rb^T-(\Tr\rb)\IdS) $ for all 
$ \rb\in\tangentS $.
As a consequence, the space of Q-tensor fields is closed by the material derivative as well as the Jaumann derivative.
Unfortunately, the upper- and lower-convected derivative fail this behavior.
Symmetric tensor fields are closed by them, but trace-free tensor fields are not, since \eqref{eq:ttensor_Dupp} and \eqref{eq:ttensor_Dlow} yield
\begin{align*}
    \Tr\Dupp\Qb = -\Tr\Dlow\Qb
        &= -2\inner{\tangentR[^2]}{ \Gbcal[\Vb_{\mfrak}], \Qb}
         = \beta\Tr\Gb[\Vb_{\mfrak}] - 2\inner{\tangentS[^2]}{ \Gb[\Vb_{\mfrak}], \qb}\formComma
\end{align*} 
which is only vanishing for rigid body motions in general.

Note that for an eigenvector field $ \Pb\in\tangentR $ with eigenvalue field $ \lambda\in\tangentS[^0] $, 
\ie\ it holds $ \Qb\Pb=\lambda\Pb $, yields
\begin{align*}
    (\Dmat\Qb)\Pb
        &= \dot{\lambda}\Pb - (\Qb - \lambda\Id)\Dmat\Pb\formComma \\
    (\Djau\Qb)\Pb
            &= \dot{\lambda}\Pb - (\Qb - \lambda\Id)\Djau\Pb
\end{align*}
by \eqref{eq:ttensor_prodwithvector_Dmat} and \eqref{eq:ttensor_prodwithvector_Djau}.
Since $ \Qb $ is a real-valued symmetric 2-tensor field, the union of all eigenspaces is spanning $ \tangentR $.
As a consequence we obtain the following corollary.
\begin{corollary} \label{cor:qtensor_kernel}
    For all $ \Qb\in\tangentQR $, $ \Pb\in\tangentR $ and $ \lambda\in\tangentS[^0] $, 
    \st\ $ \Qb\Pb=\lambda\Pb $ and $ \dot{\lambda}=0 $ is valid, holds
    \begin{align*}
        \Dmat\Pb &= 0 &&\Longrightarrow & \Dmat\Qb &= 0 \formComma \\
        \Djau\Pb &= 0 &&\Longrightarrow & \Djau\Qb &= 0 \formPeriod
    \end{align*}
\end{corollary}
The converses are not true without further ado. 
The main reason is that the eigenvector fields of a Q-tensor field do not have to be differentiable, neither spatially nor temporally.
If the Q-tensor field comprises $ \pm \frac{1}{2} $-defects, eigenvector fields even has to be discontinuous to represent such defects.
Certainly, it is feasible to show the converses by modifying the time-derivatives \wrt\ sign-sensitivity for instance.
However, for the sake simplicity, we leave this issue as an open question in this paper.
Note that corollary \ref{cor:qtensor_kernel} would not hold in the same way for the upper- and lower convected derivative.

\subsubsection{Surface Conforming Q-Tensor Fields} \label{sec:cqtensor}

One  useful subset of the space of Q-tensor fields is the space of surface conforming Q-tensor fields
$ \tangentCQR := \Qbcal[\tangentQS, 0, \tangentS[^0]]$, which is a subtensor field of the Q-tensor field space,  \ie\ $ \tangentCQR < \tangentQR < \tangentR[^2] $, see \cite{Nestler_2020, Nitschkeetal_PRSA_2020, Boucketal_arXiv_2022}.
The associated orthogonal projection $ \projCQR:\tangentQR \rightarrow \tangentCQR $ is given by
\begin{align}\label{eq:cqtensor_projection}
    \projCQR\Qb &:= \Qb - \projS(\Qb\normal)\otimes\normal - \normal\otimes\projS(\Qb\normal)
\end{align}
for all $\Qb\in\tangentQR$.
Since decomposition \eqref{eq:qtensor_decomposition} yields $ \Qb\normal = \etab + \beta\normal $,
we could summarize the situation in the following corollary.
\begin{corollary}\label{col:cqtensor}
    A Q-tensor field $ \Qb = \qb + \etab\otimes\normal + \normal\otimes\etab + \beta\left( \normal\otimes\normal - \frac{1}{2}\IdS \right) \in\tangentQR $ 
    with $ \qb\in\tangentQS$, $ \etab\in\tangentS $ and $ \beta=\tangentS[^0] $, is surface conforming,
    if and only if one of the following equivalent statements is true:
    \begin{enumerate}[(i)]
    \item $ \Qb = \projCQR\Qb \in\tangentCQR $,
    \item $ \projS(\Qb\vb) = \etab = 0 $,
    \item $ \normal $ is an eigenvector field of  $ \Qb $ and $ \beta $ is its associated eigenvalue.
    \end{enumerate}
\end{corollary}
In contrast to the Jaumann derivative \eqref{eq:qtensor_Djau}, 
the space of surface conforming Q-tensor fields is not closed by the material derivative \eqref{eq:qtensor_Dmat}, 
since $ \projS((\Dmat\Qb)\vb) = \qb\bb[\Vb_{\mfrak}] - \frac{3\beta}{2}\bb[\Vb_{\mfrak}] $ is not vanishing generally for $ \Qb\in\tangentCQR $.
To obtain such a closing we use the orthogonal projection \eqref{eq:cqtensor_projection} and call the resulting time derivative
$ \DCQmat:= \projCQR\circ\Dmat\vert_{\tangentCQR} $ surface conforming material derivative.
Taken all together, this yields
\begin{align}
    \DCQmat\Qb \label{eq:cqtensor_Dmat}
        &= \Qbcal[\dot{\qb},0,\dot{\beta}]
        &&= \dot{\qb} + \dot{\beta} \left( \normal\otimes\normal - \frac{1}{2}\IdS \right) 
        &&\in\tangentCQR\\
    \Djau\Qb \label{eq:cqtensor_Djau}
        &= \Qbcal[\timeJ\qb,0,\dot{\beta}]
        &&= \timeJ\qb + \dot{\beta} \left( \normal\otimes\normal - \frac{1}{2}\IdS \right)
        &&\in\tangentCQR
\end{align}
for all surface conforming Q-tensor fields $ \Qb = \qb + \beta (\normal\otimes\normal - \frac{1}{2}\IdS)\in\tangentCQR $,
where $ \qb\in\tangentQS $ and $ \beta\in\tangentS[^0] $.

One simple special case of conforming Q-tensor fields are the tangential Q-tensor fields in $ \tangentQS=\Qbcal[\tangentQS,0,0]<\tangentCQR $.
Here, $ \tangentQS $ is closed by the surface conforming material derivative \eqref{eq:cqtensor_Dmat}, \resp\ Jaumann derivative  \eqref{eq:cqtensor_Djau},
which coincides with the tangential material derivative, \resp\ tangential Jaumann derivative, given in \cite{NitschkeVoigt_JoGaP_2022}.

\subsubsection{Surface Landau-de Gennes models}\label{sec:LdG}

As already demonstrated in \cite{NitschkeVoigt_JoGaP_2022,NitschkeSadikVoigt_A_2022} for tangential tensor-fields the dynamics of these models differ. To sensitize the reader for this difference in applying the models, e.g. in the context of morphogenesis \cite{Maroudas-Sacks_NP_2021,Hoffmann_SA_2022,Morris_2022}, is the main motivation for this research. 

As a simple, but not trivial, example we consider the one-constant Landau-de Gennes free energy
\begin{align}\label{eq:LdG_energy}
    \potenergy[\Qb] 
        &:= \frac{L}{2} \normHsq{\tangentR[^3]}{\nablaS\Qb} + \int_{\surf} a\Tr\Qb^2 + \frac{2b}{3}\Tr\Qb^3 + c\Tr\Qb^4 \dS
\end{align}
for Q-tensor fields $ \Qb\in\tangentQR $, elastic parameter $ L>0 $ and thermotropic coefficients $ a,b,c\in\R $.
Moreover, we assume that the surface $ \surf $ is boundaryless, \ie\ $ \partial\surf=\emptyset $, 
and the motion of the surface is prescribed by the material velocity $ \Vb_{\mfrak}\in\tangentR $.
The associated $ \hil $-gradient flow with dynamics driven by 
$ \Dt\Qb\in\{ \Dmat\Qb, \Djau\Qb \} $ is
\begin{align}\label{eq:qtensor_LdG_flow}
    \Dt \Qb
        &= - \nabla_{\hilspace{\tangentQR}} \potenergy
         = L \DeltaS\Qb - 2 \left( a\Qb + b\left( \Qb^2 - \frac{\Tr(\Qb^2)}{3}\Id \right) + c \Tr(\Qb^2)\Qb \right)
         \in \tangentQR \formComma
\end{align}
where the $ \hil $-gradient $ \nabla_{\hilspace{\tangentQR}} \potenergy $ is given by 
variation of the energy  in arbitrary directions of $ \Rb\in\tangentQR $, \ie\ 
$ \innerH{\tangentQR}{\nabla_{\hilspace{\tangentQR}} \potenergy, \Rb}:= \innerH{\tangentQR}{\frac{\delta\potenergy}{\delta\Qb}, \Rb} $
with aid of lemma \ref{lem:laplace_equals_bochner}, which justifies the surface Laplace operator.
Note that $ \tangentQR $ is closed by  $\Dt=\Dmat$ as well as $ \Dt=\Djau $ in $ \tangentR[^2] $, see \eqref{eq:qtensor_Dmat} and  \eqref{eq:qtensor_Djau},
\ie\ one could safely use $ \Dmat $ and $ \Djau $ given in table \ref{tab:ttensor_timederivatives}.
For a pure tangential motion of the surface, \ie\ $\vnor=0$, and a therefore valid Eulerian observer, \ie\ $\Vb=0$, equation \eqref{eq:qtensor_LdG_flow}
equals the Q-tensor equation of the surface Beris-Edwards model in \cite{Boucketal_arXiv_2022} for the Jaumann derivative.
Note that $\DeltaS\Qb\in\tangentQR$ holds already, therefore we do not need to apply an extra projection into the space of Q-tensor fields.
This can immediately be deduced from lemma \ref{lem:laplace_equals_beltrami}.

The situation changes if we like to consider the Landau-de Gennes energy \eqref{eq:LdG_energy} \wrt\ surface conforming Q-tensor fields $ \tangentCQR $.
The associated gradient flow can be obtained either by variation and weak testing in the right space, 
\ie\ using 
$ \innerH{\tangentCQR}{ \projCQR\Dt\Qb, \Rb} = -\innerH{\tangentCQR}{\frac{\delta\potenergy}{\delta\Qb}, \Rb} $ for all $ \Rb\in\tangentCQR $,
or by Lagrange multiplier technique.
Both approaches lead to the same result as we see below.
Since $ \Qb\normal = \beta\normal $ holds for all $ \Qb\in\tangentCQR $ according to corollary \ref{col:cqtensor}, 
we infer $ \Qb^2\normal = \beta^2\normal $ and from that in turn
$ \projS\left( (\Qb^2 - \frac{\Tr(\Qb^2)}{3}\Id)\normal \right) =  0$.
Or in other words, if $ \Qb  $ is conforming, then so is the Q-tensor part of $ \Qb^2 $.
This leads to the   $ \hil $-gradient flow
\begin{align}\label{eq:qtensor_LdG_confomalflow}
    \Dt[\operatorname{C}_{\surf}] \Qb
        &= - \nabla_{\hilspace{\tangentCQR}} \potenergy
         = L \DeltaCQS\Qb - 2 \left( a\Qb + b\left( \Qb^2 - \frac{\Tr(\Qb^2)}{3}\Id \right) + c \Tr(\Qb^2)\Qb \right)
         \in \tangentCQR 
\end{align}
for $ \Qb\in\tangentCQR $, where $ \Dt[C]\Qb$ is one of the time derivatives in $\{ \DCQmat\Qb, \Djau\Qb \} $, 
$ \DeltaCQS:=\projCQR\circ\DeltaS\vert_{\tangentCQR} $ is the surface conforming Laplace operator
and $ \innerH{\tangentCQR}{\nabla_{\hilspace{\tangentQR}} \potenergy, \Rb}:= \innerH{\tangentCQR}{\frac{\delta\potenergy}{\delta\Qb}, \Rb} $
defines the surface conforming $ \hil $-gradient for all $ \Rb\in\tangentCQR $.
Alternatively, adding  the Lagrange function
$ \mathcal{C}[\Qb, \lambdab] := \innerH{\tangentS}{ \lambdab,  \projS(\Qb\vb) } $, 
where $ \lambdab\in\tangentS $ is the Lagrange multiplier,
to the  Landau-de Gennes energy \eqref{eq:LdG_energy} yields
\begin{align}
    \Dt \Qb \label{eq:tmp01}
        &= L \DeltaS\Qb - 2 \left( a\Qb + b\left( \Qb^2 - \frac{\Tr(\Qb^2)}{3}\Id \right) + c \Tr(\Qb^2)\Qb \right) 
            - \frac{1}{2} \left(\lambdab\otimes\normal + \normal\otimes\lambdab \right) \in \tangentQR \formComma\\
    0  \label{eq:tmp02}
        &= \projS(\Qb\vb) \in\tangentS\formComma
\end{align}
where $ \Dt\Qb\in\{ \Dmat\Qb, \Djau\Qb \} $.
Substituting \eqref{eq:tmp02} into \eqref{eq:tmp01} and applying $ \projCQR $ on both sides of  \eqref{eq:tmp01} also results in 
the surface conforming  $ \hil $-gradient flow \eqref{eq:qtensor_LdG_confomalflow}.
In contrast to the $ \hil $-gradient flow \eqref{eq:qtensor_LdG_flow}, the effort to rephrase the conforming flow \eqref{eq:qtensor_LdG_confomalflow} 
according to decomposition \eqref{eq:qtensor_decomposition} is significantly less.
For tangential time derivatives $ \dt\qb \in \{ \dot{\qb}, \timeJ\qb \} $ given in table \ref{tab:ttensor_tangentialtimederivatives},
we obtain the system of tangential Q-tensor and scalar equations
\begin{align}\label{eq:qtensor_LdG_confomalflow_decomposed}
   \begin{aligned}
        \dt\qb
               &= L\left(  \Delta\qb - \Tr(\shop^2)\qb + 3\beta \projQS(\shop^2) \right)
                  -\left(2a - 2b\beta + 3c\beta^2  + 2c\Tr\qb^2\right)\qb  \in \tangentQS\formComma\\
            \dot{\beta}
               &= L \left(\Delta\beta + \inner{\tangentS[^2]}{ \shop^2, 2\qb - 3\beta\IdS } \right)
                   -\left(2a + b\beta + 3c\beta^2  + 2c\Tr\qb^2\right)\beta
               +\frac{2}{3} b \Tr\qb^2 \in\tangentS[^0]\formComma
   \end{aligned}
\end{align}
which are equivalent to \eqref{eq:qtensor_LdG_confomalflow}, see \ref{sec:qtensor_LdG_confomalflow_decomposed} for a detailed derivation.
We could substitute $ \shop^2 = \meanc\shop - \gaussc\IdS $ for the third fundamental form, where $ \gaussc := \det\{ \shopC^{i}_{j} \} $ is the Gaussian curvature.
For the material derivative this yields the surface Landau-de Gennes model in  \cite{Nestler_2020} up to the uniaxiality constrain used there.
It gives also the same tangential Q-tensor equation in \cite{Nitschkeetal_PRSA_2020} for a constant $\beta$. Implications of the choice of the time derivative $ \dt\qb \in \{ \dot{\qb}, \timeJ\qb \} $ in these models needs to be explored numerically.

\appendix

\section{Identities}\label{sec:identities}

\begin{lemma}
    For a parameterization $ \para $ holds
    \begin{align}
        \partial_i\partial_j\para \label{eq:partialpartialder_X}
            &= \Gamma_{ij}^k\partial_k\para + \shopC_{ij}\normal \formComma\\
        \partial_i \normal \label{eq:partialder_normal}
            &= - \shopC^j_i \partial_j\para \\
     \text{\resp\ } 
        \inner{\tangentR}{\partial_i\partial_j\para,\partial_k\para}
            &= \Gamma_{ijk} \formComma \label{eq:christoffel_def}\\
        \inner{\tangentR}{\partial_i\partial_j\para,\normal}
           &=-\inner{\tangentR}{\partial_i\normal,\partial_j\para}
            = \shopC_{ij} \formPeriod \label{eq:shop_def}\\
         \inner{\tangentR}{\partial_i\normal,\normal} 
            &= 0 \label{eq:partial_normal_is_orth_normal}
    \end{align}
\end{lemma}
\begin{proof}
    Equation \eqref{eq:christoffel_def} is an alternative definition of the Christoffel symbols if the metric tensor is given by 
    $ g_{ij}= \inner{\tangentR}{\partial_i\para,\partial_j\para} $.
    Equations \eqref{eq:shop_def} are equivalent definitions of the second fundamental form, since \eqref{eq:partial_normal_is_orth_normal} is true,
    which in turn holds by $ \left\| \normal \right\|_{\tangentR}=1 $.
    Identities \eqref{eq:partialpartialder_X}  and \eqref{eq:partialder_normal} summarize \eqref{eq:christoffel_def}, \eqref{eq:shop_def} and \eqref{eq:partial_normal_is_orth_normal}.
\end{proof}

\begin{lemma}
    For the metric tensor holds
    \begin{align}
        \partial_l g_{ij}
            &= \Gamma_{lij} + \Gamma_{lji} \label{eq:metric_cov_partialder} \\
        \partial_l g^{ij}
            &= -\left( g^{kj}\Gamma_{lk}^i + g^{ki}\Gamma_{lk}^j \right) \formPeriod \label{eq:metric_contra_partialder}
    \end{align}
\end{lemma}
\begin{proof}
    Both identities are consequences of the metric compatibility, \ie\ 
    $ 0 = g_{ij|l}= \partial_l g_{ij} - \Gamma_{li}^k g_{kj} - \Gamma_{lj}^k g_{ik}  $ and
    $ 0 = \tensor{g}{^{ij}_{|l}} = \partial_l g^{ij} + \Gamma_{lk}^i g^{kj} + \Gamma_{lk}^j g^{ik} $.
    Alternatively, \eqref{eq:metric_cov_partialder} is a consequence of \eqref{eq:christoffel_def}, 
    and \eqref{eq:metric_contra_partialder} of \eqref{eq:metric_cov_partialder} by evaluating 
    $ \partial_l g^{ij} = \partial_l( g^{ik}g^{jm} g_{km} ) $ with the aid of the product rule.
\end{proof}

\begin{lemma}
    For a time-depending parameterization $ \para $ with velocity $ \Vb = \vb + \vnor\normal =\partial_t\para\in\tangentR $ holds
    \begin{align}
        \partial_i\Vb \label{eq:partialder_V}
            &= \tensor{[\nablaS\Vb]}{^A_i}\eb_A
             = \tensor{G}{^j_i}[\Vb]\partial_j\para + b_i[\Vb]\normal \formComma
    \end{align}
    where
    \begin{align}
        \Gb[\Vb] \label{eq:G_def}
            &:= \projS[^2](\nablaS\Vb) 
             = \nabla\vb -\vnor\shop \in\tangentS[^2]\\
        \bb[\Vb] &:= \normal \nablaS\Vb \label{eq:b_def}
                = \nabla\vnor + \shop\vb\in\tangentS \formPeriod
    \end{align}
\end{lemma}
\begin{proof}
    See \cite{NitschkeSadikVoigt_A_2022}.
\end{proof}

\begin{lemma}
    For time-dependent covariante proxy components $ r_{ij}\in\tangentS[^0] $, $ \eta_{i}\in\tangentS[^0] $, 
    contravariant proxy components $ r^{ij}\in\tangentS[^0]$ and $ \eta^{i}\in\tangentS[^0]$
    of a tangential 2-tensor fields $ \rb\in\tangentS[^2] $ and vector field $ \etab\in\tangentS $ holds
    \begin{align}
        \partial_t \eta_{i} \label{eq:vector_partialtimeder_covar}
            &= g_{ik}\partial_t \eta^k + \left[ (\Gb[\Vb]+\Gb^T[\Vb])\etab  \right]_{i} \formComma\\
        \partial_t r_{ij} \label{eq:ttensor_partialtimeder_covar}
            &= g_{ik}g_{jl}\partial_t r^{kl} + \left[ \rb(\Gb[\Vb]+\Gb^T[\Vb]) + (\Gb[\Vb]+\Gb^T[\Vb])\rb  \right]_{ij} \formComma 
    \end{align}
    with $ \Gb[\Vb]\in\tangentS[^2] $ given in \eqref{eq:G_def}.
\end{lemma}
\begin{proof}
    Follows by $ \partial_t g_{ij} = G_{ij}[\Vb] + G_{ji}[\Vb] $ (see \cite{NitschkeVoigt_JoGaP_2022}), 
    $ r_{ij}= g_{ik}g_{jl} r^{kl} $, $ \eta_{i} = g_{ik}\eta^k $ and product rule.
\end{proof}

\begin{lemma}
    For a time-depending parameterization $ \para $ with velocity $ \Vb = \vb + \vnor\normal =\partial_t\para\in\tangentR $ holds
    \begin{align}
        \partial_t\normal = \partial_t(\normalC^{A}\eb_A) \label{eq:partialtimeder_normal}
            &= -\bb[\Vb]\in\tangentS\formComma\\
        \text{\resp\ } 
          \partial_t\normalC^{A} 
            &= -b^i[\Vb]\partial_i\para^{A}\formComma \notag
    \end{align}
    where $ \bb[\Vb] $ is given in \eqref{eq:b_def}.
\end{lemma}
\begin{proof}
    Follows from 
    \begin{align*}
        \inner{\tangentR}{\partial_t\normal,\normal} 
            &= \partial_t\inner{\tangentR}{\normal,\normal} - \inner{\tangentR}{\normal,\partial_t\normal}
            = -\inner{\tangentR}{\partial_t\normal,\normal}
            \overset{\Rightarrow}{=} 0
    \end{align*}
    and 
    \begin{align*}
        \inner{\tangentR}{\partial_t\normal, \partial_i\para}
            &= \partial_t\inner{\tangentR}{\normal, \partial_i\para} - \inner{\tangentR}{\normal, \partial_i\Vb}
            \overset{\eqref{eq:partialder_V}}{=} -b_{i}[\Vb] \formPeriod
    \end{align*}
\end{proof}

\begin{corollary}
    For a time-depending parameterization $ \para $ with velocity $ \Vb = \vb + \vnor\normal =\partial_t\para\in\tangentR $ and 
    a tangential vector field $ \ub\in\tangentS $ holds
    \begin{align}
        \partial_t\normal + u^k\partial_k\normal= \partial_t(\normalC^{A}\eb_A + u^k\partial_k\normalC^{A}\eb_A)  \label{eq:timederwithu_normal}
                &= -\bb[\Vb+\ub]\in\tangentS\formComma\\
            \text{\resp\ } 
              \partial_t\normalC^{A} + u^k\partial_k\normalC^{A}
                &= -b^i[\Vb+\ub]\partial_i\para^{A}\formComma \notag
    \end{align}
    where $ \bb[\Vb+\ub] = \nabla\vnor + \shop(\vb+\ub)$ is consistent with definition \eqref{eq:b_def}.
\end{corollary}
\begin{proof}
    Follows from \eqref{eq:partialtimeder_normal} and \eqref{eq:partialder_normal}.
\end{proof}

\begin{lemma}
    For a time-depending parameterization $ \para $, with velocity $ \Vb =\partial_t\para\in\tangentR $, 
    thin film parameterization $ \chib[\para] $ and thin film velocity $ \widehat{\Vb} = \partial_t\chib[\para]\in\tangent{\surf_h}=\tangent{\R^3}\vert_{\surf_h} $
    holds 
    \begin{align}\label{eq:tfl_deformation}
        \nablahat\widehat{\Vb} \rightarrow \Gbcal[\Vb] = \Gb[\Vb] + \normal\otimes\bb[\Vb] - \bb[\Vb]\otimes\normal
    \end{align}
    for $ h \rightarrow 0 $.
\end{lemma}
\begin{proof}
    In the following, we omit the argument $ \para $ in square brackets.
    The thin film parameterization \eqref{eq:tf_para} yields the frame
    \begin{align}\label{eq:tf_frame}
        \partial_i\chib 
            &= \partial_i\para - \xi\shopC_i^j \partial_j\para \formComma
        & \partial_\xi\chib 
            &= \normal\formPeriod
    \end{align}
    Regarding this frame, the covariant thin film proxy of the velocity 
    $ \widehat{\Vb} = \partial_t\chib = \Vb - \xi\bb[\Vb] $ is given by
    \begin{align*}
        \widehat{V}_i 
            &= \inner{\tangent{\surf_h}}{\widehat{\Vb} , \partial_i\chib}
             = v_i - \xi\left( b_i[\Vb] + \shopC_{ij}v^j \right) + \landau(\xi^2) \formComma
        &\widehat{V}_\xi
            &= \inner{\tangent{\surf_h}}{\widehat{\Vb} , \partial_\xi\chib}
             = \vnor \formPeriod
    \end{align*}
    The Christoffel symbols of second kind \wrt\ the thin film frame \eqref{eq:tf_frame} are
    \begin{align*}
        \GGamma_{ij}^{k} = \Gamma_{ij}^{k} + \landau(\xi)\formComma \quad
        \GGamma_{ij}^{\xi} = \shopC_{ij} + \landau(\xi) \formComma \quad
        \GGamma_{\xi\xi}^{K} = \GGamma_{I\xi}^{\xi} = \GGamma_{\xi I}^{\xi} = 0\formComma \quad \text{and} \quad
        \GGamma_{i\xi}^{k} = \GGamma_{\xi i}^{k} &= -\shopC_{i}^{k} + \landau(\xi) \formComma
    \end{align*} 
    where a capital Latin letter $ I,J,K $ comprises a small Latin letter $ i,j,k $ and $ \xi $, see \cite{Nitschke_2018} for more details.
    Therefore the covariant thin film proxy of the velocity gradient 
    $ \nablahat\widehat{\Vb} = \delta_B^C \partial_C \widehat{V}^A \eb_A\otimes\eb_B $ yields
    \begin{align*}
        [\nablahat\widehat{\Vb}]_{ij}
            &= \partial_j \widehat{V}_i - \GGamma_{ij}^K \widehat{V}_K
             = \partial_j v_i - \Gamma_{ij}^k v_k - \vnor\shopC_{ij} + \landau(\xi)
             &&= G_{ij}[\Vb] + \landau(\xi)\\
        [\nablahat\widehat{\Vb}]_{i\xi}
            &= \partial_\xi \widehat{V}_i - \GGamma_{\xi i}^K \widehat{V}_K
             = -b_i[\Vb] - \shopC_{ij}v^j + \shopC_i^k v_k + \landau(\xi)
             &&=  -b_i[\Vb] + \landau(\xi)\\
        [\nablahat\widehat{\Vb}]_{\xi j}
            &= \partial_j \widehat{V}_\xi - \GGamma_{j\xi}^K \widehat{V}_K
             = \partial_j \vnor + \shopC_j^k v_k + \landau(\xi) 
             &&=  b_j[\Vb] + \landau(\xi) \\
        [\nablahat\widehat{\Vb}]_{\xi \xi}
            &= \partial_\xi \widehat{V}_\xi - \GGamma_{\xi\xi}^K \widehat{V}_K
             &&= 0 \formPeriod
    \end{align*}
    The orthogonality $ \partial_i\chib \bot \partial_\xi\chib  $ 
    and the thin film limit of the covariant tangential proxy of the thin film metric tensor \wrt\ frame \eqref{eq:tf_frame}, which is $ g_{ij} $, see \cite{Nitschke_2018},
    implies \eqref{eq:tfl_deformation} finally.
\end{proof}

\begin{lemma}\label{lem:laplace_equals_beltrami}
    The surface Laplace operator $ \DeltaS:\tangentR[^2]\rightarrow\tangentR[^2] $ equals the Cartesian-componentwise Laplace-Beltrami operator,
    \ie\ for all $ \Rb = R^{AB}\eb_A\otimes\eb_B\in\tangentR[^2] $ holds
    \begin{align*}
        \left[ \DeltaS \Rb \right]^{AB} 
            &= \Delta R^{AB}
             = g^{ij}\left( \partial_i\partial_j R^{AB} - \Gamma_{ij}^{k}\partial_k R^{AB} \right) \formPeriod
    \end{align*}
\end{lemma}
\begin{proof}
    Applying product rule yields
    \begin{align*}
        \left[ (\Tr\nablaS^2)\Rb \right]^{AB}
            &= \delta_{CD}g^{kl}\partial_l\left( g^{ij}\partial_j R^{AB} \partial_i \paraC^{C}  \right)\partial_k\paraC^{D}
             = g^{jl} \partial_l\partial_j R^{AB} 
                + g^{kl}g^{ij} \inner{\tangentR}{ \partial_l\partial_i \para, \partial_k \para }  \partial_j R^{AB}
                + (\partial_i g^{ij}) \partial_j R^{AB}\formPeriod
    \end{align*}
    Substituting $ g^{kl}\inner{\tangentR}{ \partial_l\partial_i \para, \partial_k \para } = \Gamma_{li}^l $ \eqref{eq:christoffel_def}
    and $\partial_i g^{ij} = -(g^{jk}\Gamma_{ik}^i + g^{ik}\Gamma_{ik}^j )$ \eqref{eq:metric_contra_partialder} 
    gives the assertion.
\end{proof}

\begin{lemma}\label{lem:laplace_equals_bochner}
    The surface Laplace operator $ \DeltaS:\tangentR[^2]\rightarrow\tangentR[^2] $ corresponds to the Bochner-like Laplace operator 
    given by the surface derivative $ \nablaS $, \ie\ for all $ \Rb\in\tangentR[^2] $ holds
    \begin{align*}
        \DeltaS\Rb = -\nablaS^{*}\nablaS\Rb \formPeriod
    \end{align*}
\end{lemma}
\begin{proof}
    Neglecting any boundary terms, lemma \ref{lem:laplace_equals_beltrami} yields
    \begin{align*}
        \innerH{\tangentR[^2]}{\DeltaS\Rb, \Psib}
            &=\innerH{\tangentR[^0]}{\Delta R^{AB}, \Psi_{AB}}
             = - \innerH{\tangentR[^1]}{\nabla R^{AB}, \nabla\Psi_{AB}}%\\
            %&= -\int_{\surf} g^{ij}\partial_i R^{AB} \partial_j \Psi_{AB} \mu
             = - \innerH{\tangentR[^3]}{\nablaS\Rb, \nablaS\Psib}
    \end{align*}
    for all $ \Rb,\Psib\in\tangentR[^2] $.
\end{proof}

\begin{corollary}\label{col:surface_laplace_decomposition}
    For all $ \Rb = \rb + \etab_L\otimes\normal + \normal\otimes\etab_R + \phi\normal\otimes\normal \in \tangentR[^2] $,
    $ \rb\in\tangentS[^2] $, $ \etab_L,\etab_R \in\tangentS $ and $ \phi\in\tangentS[^0] $,
    the surface Laplace operator $ \DeltaS:\tangentR[^2]\rightarrow\tangentR[^2] $ yields
    \begin{align}\label{eq:surface_laplace_decomposition}
        \DeltaS\Rb
            &= \Delta\rb - \left( \shop^2\rb + \rb\shop^2 \right) 
                -2\left( (\nabla\etab_{L})\shop + \shop(\nabla\etab_{R})^{T} \right) - \left( \etab_L\otimes\nabla\meanc + \nabla\meanc\otimes\etab_{R} \right)
                + 2\phi\shop^2 \notag\\
            &\quad + \left( 2(\nabla\rb):\shop + \rb\nabla\meanc 
                        + \Delta\etab_{L} - \Tr(\shop^2)\etab_{L} - \shop^{2}\left( \etab_{L} + 2\etab_{R} \right)
                        -2\shop\nabla\phi - \phi\nabla\meanc\right)\otimes\normal \notag\\
            &\quad +\normal\otimes\left( 2(\nabla\rb^T):\shop + (\nabla\meanc)\rb 
                        + \Delta\etab_{R} - \Tr(\shop^2)\etab_{R} - \shop^{2}\left( \etab_{R} + 2\etab_{L} \right)
                        -2\shop\nabla\phi - \phi\nabla\meanc \right) \notag\\
            &\quad + \left( 2 \shop^2:\rb 
                            + 2(\nabla\etab_{L}+\nabla\etab_{R}):\shop + (\etab_{L} + \etab_{R})\nabla\meanc 
                            + \Delta\phi - 2\phi \Tr(\shop^2)\right) \normal\otimes\normal \formPeriod
    \end{align}
\end{corollary}
\begin{proof}
    In this proof we calculate $ \DeltaS\Rb = \DeltaS\rb + \DeltaS(\etab_L\otimes\normal) + \DeltaS(\normal\otimes\etab_R) + \DeltaS(\phi\normal\otimes\normal) $
    term by term in this order using $ \left[ \DeltaS \Rb \right]^{AB} = g^{ij}\left( \partial_j\partial_i R^{AB} - \Gamma_{ij}^{k}\partial_k R^{AB} \right) $ (lemma \ref{lem:laplace_equals_beltrami}).
    This is a straightforward proceeding, where we mainly use $ \partial_i\partial_j\para = \Gamma_{ij}^k\partial_k\para + \shopC_{ij}\normal$ \eqref{eq:partialpartialder_X}
    and $ \partial_i \normal = - \shopC^j_i \partial_j\para $ \eqref{eq:partialder_normal}, without mentioning it every time.
    Mixed proxy components $ \tensor{[\nablaS\rb]}{^{AB}_{k}} = \partial_k r^{AB} $ yield
    \begin{align*}
        \partial_k r^{AB}
            &= \partial_k\left( r^{ij} \partial_i\paraC^{A}\partial_j\paraC^{B} \right)
             = \partial_k r^{ij} \partial_i\paraC^{A}\partial_j\paraC^{B}
                    + r^{ij}\left( \Gamma_{ki}^{l} \partial_l\paraC^{A}\partial_j\paraC^{B}  + \Gamma_{kj}^{l} \partial_i\paraC^{A}\partial_l\paraC^{B}
                                    + \shopC_{ki} \normalC^{A}\partial_j\paraC^{B} + \shopC_{kj} \partial_i\paraC^{A}\normalC^{B} \right)\\
            &= \tensor{r}{^{ij}_{|k}}  \partial_i\paraC^{A}\partial_j\paraC^{B} 
                    + r^{ij}\left( \shopC_{ki} \normalC^{A}\partial_j\paraC^{B} + \shopC_{kj} \partial_i\paraC^{A}\normalC^{B} \right) \formPeriod
    \end{align*}
    Substituting this into $ \left[ \DeltaS \rb \right]^{AB} $, the product rule gives the summands
    \begin{align*}
        g^{ij} \partial_j \left( \tensor{r}{^{lm}_{|i}}  \partial_l\paraC^{A}\partial_m\paraC^{B} \right)
            &= \left(\tensor{r}{^{lm|j}_{|j}} + g^{ij}\Gamma_{ji}^{k}\tensor{r}{^{lm}_{|k}}  \right) \partial_l\paraC^{A}\partial_m\paraC^{B}
               + \tensor{r}{^{lm}_{|j}} \left( \shopC^j_l \normalC^{A}\partial_m\paraC^{B}  + \shopC^j_m \partial_l\paraC^{A}\normalC^{B}  \right)\\
        g^{ij} \partial_j \left( r^{lm} \shopC_{il} \normalC^{A}\partial_m\paraC^{B} \right)
            &= \tensor{r}{^{lm}_{|j}}  \shopC^j_l \normalC^{A}\partial_m\paraC^{B}
              + r^{lm}\left( \shopC^j_{l|j} + g^{ij} \Gamma_{ji}^k \shopC_{kl}  \right) \normalC^{A}\partial_m\paraC^{B}
              - r^{lm} \shopC_l^j\shopC^k_j \partial_k\paraC^{A}\partial_m\paraC^{B}
              + r^{lm} \shopC_l^j\shopC_{jm} \normalC^{A}\normalC^{B} \\
        g^{ij} \partial_j \left( r^{lm} \shopC_{im} \partial_l\paraC^{A}\normalC^{B} \right)
            &= \tensor{r}{^{lm}_{|j}}  \shopC^j_m \partial_l\paraC^{A}\normalC^{B}
              + r^{lm}\left( \shopC^j_{m|j} + g^{ij} \Gamma_{ji}^k \shopC_{km}  \right) \partial_l\paraC^{A}\normalC^{B}
              - r^{lm} \shopC_m^j\shopC^k_j \partial_l\paraC^{A}\partial_k\paraC^{B}
              + r^{lm} \shopC_m^j\shopC_{jl} \normalC^{A}\normalC^{B}\\
        -g^{ij}\Gamma_{ij}^{k}\partial_k r^{AB}
            &= -g^{ij}\Gamma_{ij}^{k} \left( \tensor{r}{^{lm}_{|k}}  \partial_l\paraC^{A}\partial_m\paraC^{B} 
                                + r^{lm}\left( \shopC_{kl} \normalC^{A}\partial_m\paraC^{B} + \shopC_{km} \partial_l\paraC^{A}\normalC^{B} \right) \right) \formComma
    \end{align*}
    which are adding up to
    \begin{align*}
        \DeltaS \rb 
            &= \Delta\rb - \left( \shop^2\rb + \rb\shop^2 \right) 
                + \left( 2(\nabla\rb):\shop + \rb\nabla\meanc  \right) \otimes \normal
                + \normal\otimes\left( 2(\nabla\rb^T):\shop + (\nabla\meanc)\rb \right)
                + 2 (\shop^2:\rb) \normal\otimes\normal \formComma 
    \end{align*}
    where we use that $ \shopC^j_{i|j} = \shopC^j_{j|i} = \meanc_{|i} $ is valid, since $ \shop $ is curl-free.
    Mixed proxy components $ \tensor{[\nablaS(\etab_{L}\otimes\normal)]}{^{AB}_{k}} = \partial_k (\eta_L^A\normalC^B) $ yield
    \begin{align*}
        \partial_k (\eta_L^A\normalC^B)
            &= \partial_k (\eta_L^i \partial_i\paraC^{A}\normalC^B)
             = \eta_{L|k}^i\partial_i\paraC^{A}\normalC^B
               + \eta_L^i \shopC_{ki}\normalC^A\normalC^B  
               - \eta_L^i \shopC^j_k \partial_i\paraC^{A} \partial_j\paraC^{B} \formPeriod
    \end{align*}
    Substituting this into $ \left[ \DeltaS (\etab_{L}\otimes\normal) \right]^{AB} $, the product rule gives the summands
    \begin{align*}
        g^{ij} \partial_j\left( \eta_{L|i}^k\partial_k\paraC^{A}\normalC^B \right)
            &= \left( \eta_{L\ |j}^{k|j} + g^{ij}\Gamma_{ij}^l \eta_{L|l}^k \right)\partial_k\paraC^{A}\normalC^B
               + \eta_{L|i}^k \shopC^i_k \normalC^A\normalC^B
               - \eta_{L|i}^k \shopC^{il} \partial_k\paraC^{A} \partial_l\paraC^{B}\\
        g^{ij} \partial_j\left( \eta_L^k \shopC_{ik}\normalC^A\normalC^B \right)
            &=   \eta_{L|j}^k  \shopC^j_k \normalC^A\normalC^B 
               + \eta_L^k \left( \shopC^j_{k|j} + g^{ij}\Gamma_{ji}^l\shopC_{lk}  \right)\normalC^A\normalC^B
               - \eta_L^k \shopC^j_k \shopC^l_j \left( \partial_l\paraC^A \normalC^B - \normalC^A\partial_l \paraC \right)\\
        -g^{ij} \partial_j \left( \eta_L^k \shopC^l_i \partial_k\paraC^{A} \partial_l\paraC^{B} \right)
            &= - \eta_{L|j}^k \shopC^{jl} \partial_k\paraC^{A} \partial_l\paraC^{B}
               - \eta_L^k \left( \tensor{\shopC}{^{lj}_{|j}} + g^{ij}\Gamma_{ji}^m\shopC^l_m \right)\partial_k\paraC^{A} \partial_l\paraC^{B}
               - \eta_L^k \shopC^l_i \left( \shopC^i_k \normalC^A\partial_l\paraC^B + \shopC^i_l \partial_k\paraC^A \normalC^{B} \right)\\
        -g^{ij}\Gamma_{ij}^{k}\partial_k (\eta_L^A\normalC^B)
            &= -g^{ij}\Gamma_{ij}^{k} 
                \left(\eta_{L|k}^l\partial_l\paraC^{A}\normalC^B
                   + \eta_L^l \shopC_{kl}\normalC^A\normalC^B  
                   - \eta_L^l \shopC^m_k \partial_l\paraC^{A} \partial_m\paraC^{B}\right) \formComma
    \end{align*}
    which are adding up to
    \begin{align*}
        \DeltaS(\etab_L\otimes\normal)
            &= -2(\nabla\etab_{L})\shop - \etab_{L}\otimes\nabla\meanc
               +\left( \Delta\etab_{L} - \Tr(\shop^2)\etab_{L} - \shop^{2} \etab_{L} \right)\otimes\normal\\
            &\quad - 2 \normal\otimes \shop^{2} \etab_{L}
               +\left( 2(\nabla\etab_{L}):\shop + \etab_{L}\nabla\meanc \right) \normal\otimes\normal \formPeriod
    \end{align*}
    Since $ \DeltaS $ is compatible with transposition, 
    \ie\ it is $ \DeltaS ( \normal \otimes \etab_R ) = (\DeltaS(\etab_R\otimes\normal) )^T $ valid, this leads to
    \begin{align*}
        \DeltaS ( \normal \otimes \etab_R )
            &= -2\shop(\nabla\etab_{R})^T - (\nabla\meanc) \otimes\etab_{R}
               +  \normal \otimes\left( \Delta\etab_{R} - \Tr(\shop^2)\etab_{R} - \shop^{2} \etab_{R} \right) \\
            &\quad - 2  \shop^{2} \etab_{R} \otimes \normal
               +\left( 2(\nabla\etab_{R}):\shop + \etab_{R}\nabla\meanc \right) \normal\otimes\normal \formPeriod
    \end{align*}
    Mixed proxy components $ \tensor{[\nablaS(\phi\normal\otimes\normal)]}{^{AB}_{k}} = \partial_k (\phi\normalC^A\normalC^B) $ yield
    \begin{align*}
       \partial_k (\phi\normalC^A\normalC^B)
            &= \phi_{|k}  \normalC^A\normalC^B - \phi\shopC^l_k\left(  \partial_l\paraC^A\normalC^B +  \normalC^A\partial_l\paraC^B \right)
    \end{align*}
     Substituting this into $ \left[ \DeltaS (\phi\normal\otimes\normal) \right]^{AB} $, the product rule gives the summands
     \begin{align*}
        g^{ij} \partial_j \left( \phi_{|i}  \normalC^A\normalC^B \right)
            &= \left( \tensor{\phi}{^{|j}_{|j}} + g^{ij}\Gamma_{ij}^k\phi_{|k} \right)\normalC^A\normalC^B
                - \phi_{|i} \shopC^{il} \left( \partial_l\paraC^A\normalC^B +  \normalC^A\partial_l\paraC^B \right)\\
         -g^{ij} \partial_j \left( \phi\shopC^l_i \partial_l\paraC^A\normalC^B \right)
            &= -\phi_{|j} \shopC^{jl} \partial_l\paraC^A\normalC^B
               -\phi\left( \tensor{\shopC}{^{lj}_{|j}} + g^{ij}\Gamma_{ij}^k\shopC^l_k \right)\partial_l\paraC^A\normalC^B
               -\phi \shopC^l_i \shopC^i_l \normalC^A\normalC^B
               + \phi \shopC^l_i \shopC^{ik} \partial_l\paraC^A \partial_k\paraC^B \\
        -g^{ij} \partial_j \left( \phi\shopC^l_i \normalC^A\partial_l\paraC^B \right) 
            &= -\phi_{|j} \shopC^{jl} \normalC^A\partial_l\paraC^B
               - \phi\left( \tensor{\shopC}{^{jl}_{|j}} + g^{ij}\Gamma_{ij}^k\shopC^l_k \right)\normalC^A\partial_l\paraC^B
               - \phi \shopC^l_i \shopC^i_l \normalC^A\normalC^B
               + \phi \shopC^l_i \shopC^{ik} \partial_k\paraC^A \partial_l\paraC^B\\
        -g^{ij}\Gamma_{ij}^{k}\partial_k (\phi\normalC^A\normalC^B)
            &= -g^{ij}\Gamma_{ij}^{k}\left( \phi_{|k}  \normalC^A\normalC^B 
                                - \phi\shopC^l_k\left(  \partial_l\paraC^A\normalC^B +  \normalC^A\partial_l\paraC^B \right) \right) \formComma
     \end{align*}
      which are adding up to
      \begin{align*}
        \nablaS(\phi\normal\otimes\normal)
            &= 2\phi\shop^2
                - \left( 2\shop\nabla\phi + \phi\nabla\meanc \right)\otimes\normal
                - \normal\otimes \left( 2\shop\nabla\phi + \phi\nabla\meanc \right)
                + \left( \Delta\phi - 2\phi \Tr(\shop^2) \right) \normal\otimes\normal \formPeriod
      \end{align*}
\end{proof}

\begin{lemma}\label{lem:ssq}
    For all symmetric tangential 2-tensor fields $\sbb\in\tangentSymS:=\{ \rb\in\tangentS[^2] \,\vert\, \rb=\rb^T \} $ and tangential Q-tensor fields $\qb\in\tangentQS$ holds
    \begin{align*}
        \projQS(\sbb^2\qb) 
            &= \frac{1}{2} \normsq{\tangentSymS}{\sbb} \qb \formPeriod
    \end{align*}
\end{lemma}
\begin{proof}
    We use the Levi-Civita tensor $\Eb\in\tangentAS:=\{ \rb\in\tangentS[^2] \,\vert\, \rb=-\rb^T \}$.
    It is a skew-symmetric tangential 2-tensor field defined by its covariant proxy components $E_{ij} := \sqrt{\det\gb}\varepsilon_{ij}$,
    where $\{ \varepsilon_{ij} \}$ are the Levi-Civita symbols, see \cite{Nitschke_2018, NitschkeSadikVoigt_A_2022} for more details.
    This tensor field is very useful in many situations involving tangential tensor fields. 
    We use the properties that $\rb\bot(\rb\Eb)$ is valid for all $\rb\in\tangentS[^2]$ and $\qb\Eb\in\tangentQS$.
    This yields
    \begin{align*}
        \inner{\tangentQS}{\projQS(\sbb^2\qb), \qb\Eb}
            &= \inner{\tangentS[^2]}{\sbb^2\qb, \qb\Eb}
             = \inner{\tangentS[^2]}{\sbb\qb, (\sbb\qb)\Eb}
             = 0 \formPeriod
    \end{align*}
    Since $ \qb^2= \frac{\Tr\qb^2}{2}\IdS $ is valid, see \cite[Cor~A.4.]{Nitschke_2018}, we obtain
    \begin{align*}
        \inner{\tangentQS}{\projQS(\sbb^2\qb), \qb}
            &= \inner{\tangentS[^2]}{\sbb^2\qb, \qb}
             = \inner{\tangentS[^2]}{\sbb^2, \qb^2}
             = \frac{1}{2} \Tr\sbb^2\Tr\qb^2 \formPeriod
    \end{align*}
    Assuming $\qb\neq 0$ everywhere without loss of generality, we can span the space of Q-tensor fields by $\tangentQS = \operatorname{Span}_{\tangentS[^0]}\{ \qb , \qb\Eb\}$.
    Due to this we get
    \begin{align*}
        \projQS(\sbb^2\qb) 
            &= \frac{\inner{\tangentQS}{\projQS(\sbb^2\qb), \qb}}{\normsq{\tangentQS}{\qb}}\qb
               + \frac{\inner{\tangentQS}{\projQS(\sbb^2\qb), \qb\Eb}}{\normsq{\tangentQS}{\qb\Eb}}\qb\Eb
             = \frac{1}{2} \normsq{\tangentSymS}{\sbb} \qb \formComma
    \end{align*}
    since $\Tr\sbb^2 = \normsq{\tangentSymS}{\sbb}$ for all $\sbb\in\tangentSymS$.
\end{proof}

\section{Outsourced Calculations}

\subsection{Time Derivative on Scalar Fields}\label{sec:scalar_timeder_scratch}

Local observer coordinate parameters $ (y_{\ofrak}^1,y_{\ofrak}^2) $ can be given by
\begin{align*}
    y_{\ofrak}^i 
        &= y_{\ofrak}^i(t, y_{\mfrak}^1, y_{\mfrak}^2)
        = (\para_{\ofrak}\vert_{t}^{-1}\circ\para_{\mfrak})(t,y_{\mfrak}^1,y_{\mfrak}^2)
\end{align*}
depended on local material coordinate parameters $ (y_{\mfrak}^1,y_{\mfrak}^2) $  at time $ t $.
Therefore, with relation \eqref{eq:tensor_observer_relation}, a scalar field $ f[\para_\mfrak]\in\tangentS[^0] $ and the pullback \eqref{eq:scalar_pullback} yields
\begin{align*}
    f[\para_{\mfrak}](t,y_{\mfrak}^1, y_{\mfrak}^2)
        &= f[\para_{\ofrak}](t, y_{\ofrak}^1(t, y_{\mfrak}^1, y_{\mfrak}^2),y_{\ofrak}^2(t, y_{\mfrak}^1, y_{\mfrak}^2))\\
    (\Phi^{*_0}_{t,\tau}f[\para_{\mfrak}]\vert_{t+\tau})(t,y_{\mfrak}^1,y_{\mfrak}^2) 
        &= f[\para_{\ofrak}](t+\tau, y_{\ofrak}^1(t+\tau, y_{\mfrak}^1, y_{\mfrak}^2),y_{\ofrak}^2(t+\tau, y_{\mfrak}^1, y_{\mfrak}^2)) \formPeriod
\end{align*}
Taylor expansion of the pullback at $ \tau=0 $ gives
\begin{align}
    \MoveEqLeft f[\para_{\ofrak}](t+\tau, y_{\ofrak}^1(t+\tau, y_{\mfrak}^1, y_{\mfrak}^2),y_{\ofrak}^2(t+\tau, y_{\mfrak}^1, y_{\mfrak}^2)) \notag\\
        &= f[\para_{\ofrak}](t, y_{\ofrak}^1(t, y_{\mfrak}^1, y_{\mfrak}^2),y_{\ofrak}^2(t, y_{\mfrak}^1, y_{\mfrak}^2))
            + \tau \partial_t f[\para_{\ofrak}](t, y_{\ofrak}^1(t, y_{\mfrak}^1, y_{\mfrak}^2),y_{\ofrak}^2(t, y_{\mfrak}^1, y_{\mfrak}^2)) \notag\\
         &\quad + \tau \partial_t y_{\ofrak}^i(t, y_{\mfrak}^1, y_{\mfrak}^2) 
                        \partial_i f[\para_{\ofrak}](t, y_{\ofrak}^1(t, y_{\mfrak}^1, y_{\mfrak}^2),y_{\ofrak}^2(t, y_{\mfrak}^1, y_{\mfrak}^2))
            +\landau(\tau^2) \formPeriod
  \label{eq:eq0001}
\end{align}
To express $\partial_t y_{\ofrak}^i(t, y_{\mfrak}^1, y_{\mfrak}^2)$ also in terms of $y_{\ofrak}^i(t, y_{\mfrak}^1, y_{\mfrak}^2)$ we calculate
\begin{align*}
    \MoveEqLeft 
     \partial_t y_{\ofrak}^i(t, y_{\mfrak}^1, y_{\mfrak}^2) \partial_i \para_{\ofrak}(t, y_{\ofrak}^1(t, y_{\mfrak}^1, y_{\mfrak}^2),y_{\ofrak}^2(t, y_{\mfrak}^1, y_{\mfrak}^2))\\
        &= \frac{d}{dt} \para_{\ofrak}(t, y_{\ofrak}^1(t, y_{\mfrak}^1, y_{\mfrak}^2),y_{\ofrak}^2(t, y_{\mfrak}^1, y_{\mfrak}^2)) 
                - \partial_t \para_{\ofrak}(t, y_{\ofrak}^1(t, y_{\mfrak}^1, y_{\mfrak}^2),y_{\ofrak}^2(t, y_{\mfrak}^1, y_{\mfrak}^2)) \\
        &= \partial_t \para_{\mfrak}(t,y_{\mfrak}^1, y_{\mfrak}^2) 
            - \partial_t \para_{\ofrak}(t, y_{\ofrak}^1(t, y_{\mfrak}^1, y_{\mfrak}^2),y_{\ofrak}^2(t, y_{\mfrak}^1, y_{\mfrak}^2))\\
        &= \Vb_{\mfrak}[\para_\mfrak](t,y_{\mfrak}^1, y_{\mfrak}^2) 
            - \Vb_{\ofrak}[\para_\ofrak](t, y_{\ofrak}^1(t, y_{\mfrak}^1, y_{\mfrak}^2),y_{\ofrak}^2(t, y_{\mfrak}^1, y_{\mfrak}^2))\\
        &= \ub[\para_\ofrak, \para_\mfrak](t, y_{\ofrak}^1(t, y_{\mfrak}^1, y_{\mfrak}^2),y_{\ofrak}^2(t, y_{\mfrak}^1, y_{\mfrak}^2)) \formComma
\end{align*}
\ie\  it holds 
$ \partial_t y_{\ofrak}^i(t, y_{\mfrak}^1, y_{\mfrak}^2)
        = u^i[\para_\ofrak, \para_\mfrak](t, y_{\ofrak}^1(t, y_{\mfrak}^1, y_{\mfrak}^2),y_{\ofrak}^2(t, y_{\mfrak}^1, y_{\mfrak}^2)) $
\wrt\ the observer frame induced by $ \para_{\ofrak} $.
With $ \dot{f}:= \Dt\vert_{\Phi^{*}_{t,\tau}=\Phi^{*_0}_{t,\tau}}f $ and Taylor expansion \eqref{eq:eq0001}, 
the time derivative \eqref{eq:tensor_timeD_general} becomes
\begin{align*}
    \dot{f}[\para_{\mfrak}](t,y_{\mfrak}^1, y_{\mfrak}^2)
     &=\dot{f}[\para_{\ofrak}](t, y_{\ofrak}^1(t, y_{\mfrak}^1, y_{\mfrak}^2),y_{\ofrak}^2(t, y_{\mfrak}^1, y_{\mfrak}^2))\\
        &= \partial_t f [\para_{\ofrak}](t, y_{\ofrak}^1(t, y_{\mfrak}^1, y_{\mfrak}^2),y_{\ofrak}^2(t, y_{\mfrak}^1, y_{\mfrak}^2))\\
        &\quad + u^i[\para_\ofrak, \para_\mfrak](t, y_{\ofrak}^1(t, y_{\mfrak}^1, y_{\mfrak}^2),y_{\ofrak}^2(t, y_{\mfrak}^1, y_{\mfrak}^2))
                \partial_i f[\para_{\ofrak}](t, y_{\ofrak}^1(t, y_{\mfrak}^1, y_{\mfrak}^2),y_{\ofrak}^2(t, y_{\mfrak}^1, y_{\mfrak}^2))\formPeriod
\end{align*}
Since none of the terms need to use the  local material coordinate parameters $ (y_{\mfrak}^1,y_{\mfrak}^2) $ at time $ t $ anymore, we get
\begin{align*}
    \dot{f}[\para_{\ofrak}](t, y_{\ofrak}^1,y_{\ofrak}^2)
        &= \partial_t f [\para_{\ofrak}](t, y_{\ofrak}^1,y_{\ofrak}^2) 
                + \nabla_{\ub[\para_\ofrak, \para_\mfrak](t, y_{\ofrak}^1,y_{\ofrak}^2) } f[\para_{\ofrak}](t, y_{\ofrak}^1,y_{\ofrak}^2)
\end{align*}
finally by local evaluations at events $ (t, y_{\ofrak}^1,y_{\ofrak}^2) $ instead of $ (t, y_{\mfrak}^1,y_{\mfrak}^2) $.

\subsection{Decomposition of the Surface conforming $\hil$-Gradient Flow \eqref{eq:qtensor_LdG_confomalflow}}\label{sec:qtensor_LdG_confomalflow_decomposed}

In this section we are assuming the orthogonal surface conforming decomposition $ \Qb=\qb + \beta(\normal\otimes\normal - \frac{1}{2}\IdS) \in \tangentCQR $,
with uniquely given $ \qb\in\tangentQS $ and $\beta\in\tangentS[^0]$, 
\ie\ it is $ \Qb=\Qbcal\left[\qb,0,\beta\right] $ valid according to Q-tensor decomposition \eqref{eq:qtensor_decomposition}.
The decomposition of the left-handed side of \eqref{eq:qtensor_LdG_confomalflow} is already clarified with 
\eqref{eq:cqtensor_Dmat} for the surface conforming material derivative and \eqref{eq:cqtensor_Djau} for the Jaumann derivative.
For the elastic part we use the decomposition of the Laplace operator in corollary \ref{col:surface_laplace_decomposition} with $ \rb= \qb-\frac{\beta}{2}\IdS $, 
$ \etab_{L}=\etab_{R}=0 $ and $ \phi=\beta $.
Hence with surface conforming projection \eqref{eq:cqtensor_projection}, this yields
\begin{align*}
    \DeltaCQS\Qb
        &= \projCQR \DeltaS\Qb
         = \Delta\qb - \frac{\Delta\beta}{2}\IdS - \left( \shop^2\qb + \qb\shop^2 \right) 
            + 3\beta\shop^2
            + \left( 2 \shop^2:\qb + \Delta\beta - 3\beta \Tr(\shop^2)\right) \normal\otimes\normal\\
        &=  \Delta\qb - \projQS\left( \shop^2 \left( 2\qb - 3\beta\IdS \right) \right)
            + \left( \Delta\beta + \inner{\tangentS[^2]}{ \shop^2, 2\qb - 3\beta\IdS } \right) \left( \normal\otimes\normal - \frac{1}{2}\IdS \right) \\
        &= \Delta\qb - \Tr(\shop^2)\qb + 3\beta \projQS(\shop^2)
            + \left( \Delta\beta + \inner{\tangentS[^2]}{ \shop^2, 2\qb - 3\beta\IdS } \right) \left( \normal\otimes\normal - \frac{1}{2}\IdS \right)
            \in\tangentCQR \formPeriod
\end{align*}
The bottom line follows from $ 2\projQS( \shop^2\qb ) = \Tr(\shop^2)\qb $ (lemma \ref{lem:ssq}).
For the thermotropic part we first calculate
\begin{align*}
    \Qb^2 &= \left( \qb + \beta(\normal\otimes\normal - \frac{1}{2}\IdS) \right)^2
        = \qb^2 - \beta\qb + \beta^2\left( \normal\otimes\normal + \frac{1}{4}\IdS \right)\formPeriod
\end{align*}
Its trace and Q-tensor part is
\begin{align*}
    \Tr\Qb^2
        &= \Tr\qb^2 + \frac{3}{2}\beta^2\\
    \Qb^2 - \frac{\Tr(\Qb^2)}{3}\left( \IdS + \normal\otimes\normal \right)
        &= \qb^2 - \beta\qb - \left( \frac{\beta^2}{4} + \frac{\Tr\qb^2}{3} \right)\IdS
                  + \left( \frac{\beta^2}{2} - \frac{\Tr\qb^2}{3} \right) \normal\otimes\normal\\
        &= - \beta\qb +\left( \frac{\beta^2}{2} - \frac{\Tr\qb^2}{3} \right) \left( \normal\otimes\normal - \frac{1}{2}\IdS \right)
        \in\tangentCQR\formComma
\end{align*}
where we use that $ \Id=\IdS +  \normal\otimes\normal $  and $ \qb^2= \frac{\Tr\qb^2}{2}\IdS $ \cite[Cor~A.4.]{Nitschke_2018} hold.
Eventually, orthogonality $ ( \normal\otimes\normal - \frac{1}{2}\IdS)\bot\tangentQS $ results in the 
decomposed surface conforming $ \hil $-gradient flow \eqref{eq:qtensor_LdG_confomalflow_decomposed}.

\vspace*{1cm}
\noindent
{\bf Acknowledgements}: AV was supported by DFG through FOR3013. 

\bibliography{bib}

\end{document}